\newcommand{\CLASSINPUTtoptextmargin}{0.9in}
\newcommand{\CLASSINPUTbottomtextmargin}{0.9in}
\documentclass[12pt, draftclsnofoot, onecolumn]{IEEEtran}
\ifCLASSINFOpdf
\else
\fi
\makeatletter
\long\def\@makecaption#1#2{\ifx\@captype\@IEEEtablestring%
\footnotesize\begin{center}{\normalfont\footnotesize #1}\\
{\normalfont\footnotesize\scshape #2}\end{center}%
\@IEEEtablecaptionsepspace
\else
\@IEEEfigurecaptionsepspace
\setbox\@tempboxa\hbox{\normalfont\footnotesize {#1.}~~ #2}%
\ifdim \wd\@tempboxa >\hsize%
\setbox\@tempboxa\hbox{\normalfont\footnotesize {#1.}~~ }%
\parbox[t]{\hsize}{\normalfont\footnotesize \noindent\unhbox\@tempboxa#2}%
\else
\hbox to\hsize{\normalfont\footnotesize\hfil\box\@tempboxa\hfil}\fi\fi}
\makeatother
\usepackage{graphicx,cite,epsfig,amssymb,amsmath,algorithm,algorithmic,color}
\usepackage{ulem}
\usepackage{algorithm}
\usepackage{algorithmic}

\newcommand{\Cut}[1]{}
\newtheorem{theorem}{Theorem}
\newtheorem{lemma}{Lemma}

\usepackage{amsmath}
\DeclareMathOperator*{\argmax}{arg\,max}
\usepackage{epstopdf}
\usepackage{bm}
\usepackage{cases}

\begin{document}
%
\title{UAV-assisted Cooperative Communications with Wireless Information and Power Transfer}
%
%
%

\author{\normalsize Sixing Yin, Jing Tan and Lihua Li\\
Beijing University of Posts and Telecommunications \\  \{yinsixing, tanjing, lilihua\}@bupt.edu.cn\\
}

\maketitle

\begin{abstract}
In this paper, we focus on a typical cooperative communication system with one pair of source and destination, where a unmanned aerial vehicle (UAV) flying from a start location to an end location serves as a mobile relay. To efficiently utilize energy in ambient environment, the UAV's transmission capability is powered exclusively by radio signal transmitted from the source via the power-splitting mechanism. In such a cooperative communication system, we study the end-to-end cooperative throughput maximization problem by optimizing the UAV's power profile, power-splitting ratio profile and trajectory for both amplify-and-forward (AF) and decode-and-forward (DF) protocols. The problem is decomposed into two subproblems: profile optimization given trajectory and trajectory optimization given profile. The former one is solved via the dual decomposition with in-depth analysis and the latter one is solved via successive convex optimization, by which a lower bound is iteratively maximized. Then the end-to-end cooperative throughput maximization problem is solved by alternately solving the two subproblems. The numerical results show that with the proposed optimal solution, choice for the UAV's power profile and power-splitting ratio profile is more long-sighted than the greedy strategy from our previous work and the successive optimization is able to converge in a few rounds of iteration. Moreover, as for the end-to-end cooperative throughput, the proposed optimal solution outperforms both static and greedy strategies, especially for the AF protocol.
\end{abstract}


%
\IEEEpeerreviewmaketitle

\section{Introduction}
In recent years, wireless communications aided by unmanned aerial vehicles (UAV, also known as drones) have received significant attention from academia, industry as well as government\cite{5892898}. Compared to  terrestrial infrastructures in conventional wireless communications, UAVs that serve as aerial transceivers have remarkable advantages in terms of low cost, miniaturized size, high mobility and deployment flexibility such that UAV-assisted wireless communication systems have been extensively applied in a number of applications such as military operations and scientific missions.

UAV-aided ubiquitous coverage and UAV-aided information dissemination are two major applications of UAV-assisted communications\cite{7470933}. In the former one, UAVs are deployed to assist existing communication infrastructures in order to provide seamless coverage within an area. Two typical scenarios are rapid service recovery (e.g., when infrastructures are damaged due to natural disasters\cite{6766080}) and base station offloading in extremely crowded areas, which is also one of the five key scenarios that need to be effectively addressed in the fifth generation (5G) wireless communication systems\cite{6815890}. In the latter one, UAVs are deployed to disseminate (collect) delay-tolerant information to (from) a large number of wireless devices. A typical example is information collection in smart agriculture applications enabled by wireless sensor networks.

Another important application for UAV-assisted wireless communications is to serve as relays in cooperative communication systems, which is an effective technique for communication performance improvement and coverage range expansion under weak channel condition (due to long distance or severe obstacles between source and destination). Due to their high mobility, while serving as aerial mobile relays, UAVs are more likely to find better locations and gain more favorable channel condition (e.g., better chance of line-of-sight links) by dynamic location adjustment such that the cooperative communication performance can be significantly enhanced even when the direct link between source and destination is severely blocked, especially for delay-tolerant applications. 

Despite a large number of potential applications, UAV-assisted wireless communications are still facing new challenges. One of such practical issues is that performance and lifetime of the UAVs are usually limited by the onboard battery, which is designed as low-capacity one to cater for aircraft's size and weight. To tackle such an issue, there has been research work focusing on energy efficiency enhancement for UAV-assisted wireless communication systems\cite{7888557}\cite{7997208}. From a different perspective, using renewable energy as replenishment for UAV-assisted wireless communication systems could be another option. However, mounting UAVs with off-the-self energy harvesting devices (e.g., solar panels, which are far bigger than UAVs themselves) might significantly increase their weight and lead to even lower energy efficiency. Recently, simultaneous wireless information and power transfer (SWIPT) receives extensive attention since it fully utilizes not only information but also energy carried by radio signal. As an upsurge of recent research topics, there have been volumes of literatures on SWIPT in terms of theoretical analysis\cite{4595260}\cite{5513714}, practical implementation\cite{6567869}\cite{6804407} as well as applications in wireless communication systems\cite{6661329}\cite{6589954}. Therefore, SWIPT is a cost-effective way to replenish a UAV's energy without additional energy harvesting devices installed. 

Based on such motivation, in this paper, we consider a typical cooperative communication system with one source and one destination, which is assisted by a UAV serving as an aerial mobile relay with both amplify-and-forward (AF) and decode-and-forward (DF) protocols. To efficiently utilize energy in ambient environment, we propose to introduce SWIPT into the UAV-assisted cooperative communication system, i.e., the UAV's transmission capability is powered by radio signal transmitted from the source while the onboard battery accounts only for its maneuverability. For SWIPT at the UAV, we focus on the power-splitting mechanism, through which, radio signal received at the UAV is split into two power streams, one of which is used for energy harvesting and the other for information relaying (quantitatively controlled by power-splitting ratio). We study the cooperative throughput maximization problem by taking into all the important design aspects for the UAV's relaying power, SWIPT scheme and route planning. The main contributions of this paper are summarized as follows:

\begin{enumerate}
\item We propose a basic model for a UAV-assisted cooperative communication system with SWIPT, where a UAV serves as an aerial mobile relay and its transmission capability is powered exclusively by radio signal sent from the source via the power-splitting scheme. The cooperative throughput maximization problem is then formulated to optimize the UAV's power profile, power-splitting ratio profile and trajectory, and is solved via alternate optimizing two subproblems.

\item With the UAV's trajectory fixed, we first prove the convexity of the first subproblem and derive the closed-form solutions to the UAV's optimal power profile and power-splitting ratio profile given the dual variables with in-depth analysis for both AF and DF protocols. With such closed-form solutions, the dual variables can then be readily derived via subgradient-based methods.

\item With the UAV's power profile and power-splitting ratio profile fixed, we resort to successive convex optimization technique to optimize the UAV's trajectory, which transforms the subproblem into an incremental trajectory optimization problem, and a lower bound for the original objective function of the subproblem is iteratively maximized. Then the original cooperative throughput maximization problem is solved via alternately solving the two subproblems.

\end{enumerate}

Notably, a similar scenario can be found in one of the pioneer work in UAV-assisted mobile relaying systems, where joint optimization for the UAV's power profile and trajectory has been investigated to maximize the end-to-end cooperative throughput\cite{7572068}. Differently, in this paper, we assume that the UAV's transmission capability is powered by a ``free lunch'', i.e., energy carried by radio signal from the source. Such a problem is more complicated since in addition to optimizing the UAV's power profile and trajectory, which has been investigated in \cite{7572068}, the UAV's power-splitting ratio profile must be taken into account as well. Besides, it is also noteworthy that in this paper, we focus on a more general cooperative communication system with both AF and DF protocols, which involves the direct link between the source and the destination, while only a relay system without the direct link is considered in \cite{7572068}.

The rest of this paper is organized as follows. Related work is reviewed in Section II. In Section III, the UAV-assisted cooperative communication system is described in detail and the end-to-end cooperative throughput optimization problem is mathematically formulated and decomposed into two subproblems, which in turn, are investigated in Section IV for profile optimization with fixed trajectory and Section V for trajectory optimization with fixed profile. In Section VI, an iterative algorithm is proposed to jointly optimize  the UAV's power profile, power-splitting ratio profile and trajectory based on the analytical results in Section IV and V. Numerical results are presented in Section VII to compare the proposed optimal solution with the greedy strategy from our previous work as well as the static strategy. Finally, we conclude this paper in Section VIII.

\section{Related Work}

For UAV-assisted wireless networks, there has been a great deal of research work on the UAV's placement and deployment, where the UAVs serve as static base stations (also known as low-altitude platforms) located in the air to support ground users within an area poor in wireless connectivity. In \cite{6863654}, the optimal altitude of a single low-altitude UAV is investigated to provide maximum coverage for ground users. The authors in \cite{7486987} focus on deployment for multiple UAVs to enhance coverage by considering directional antennas for the UAVs. The authors in \cite{7762053} study a similar problem from a different angle, where the number of UAV-mounted mobile base stations is minimized while ensuring that all the ground terminals can lie within the communication range of at least one mobile base station. Besides, \cite{7510820} explores three-dimensional space and investigates the placement problem for aerial base stations to maximize the revenue, which is defined by the maximum number of users covered by a drone-cell. In general, research work in this line focuses on wireless networks assisted by static UAVs without taking their high mobility into account.

To exploit the UAV’s high mobility, flexibility and manoeuvrability in order to enhance communication performance, more work on the UAV's trajectory design as well as user scheduling can be recently found. In \cite{7572068}, the authors investigate the throughput maximization problem in a relay system served by a UAV and source/relay transmit power and the UAV's trajectory are jointly optimized. Instead of throughput, \cite{7888557} focuses on the UAV's trajectory design to maximize energy efficiency by taking the UAV's propulsion energy consumption into account. The authors in \cite{DBLP:journals/corr/WuZZ17b} extend the trajectory design problem to a wireless network with multiple UAVs, where the UAV's trajectory and power as well as user association are jointly optimized. In \cite{7556368}, a cyclical pattern is proposed for the UAV's trajectory to exploit periodic channel variations and a cyclical multiple access scheme is presented for ground user scheduling to maximize the minimum throughput, which reveals a fundamental tradeoff between throughput and access delay. 

\begin{figure}
\centering
 \includegraphics[width=4.0in]{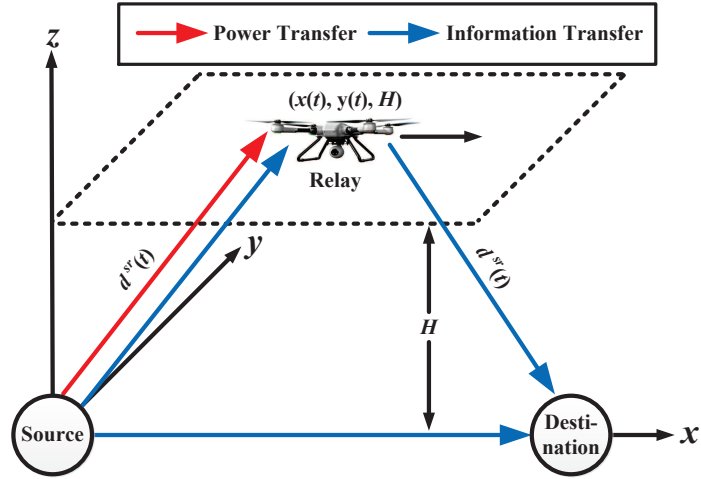}
\caption{A UAV-assisted cooperative communication system, where the UAV's transmission capability is exclusively powered by radio signal transmitted from the source.}
\label{fig:scenario}
\vspace*{-0.15in}
\end{figure}

\section{System Model and Problem Formulation}
As shown in Fig. \ref{fig:scenario}, we consider a point-to-point wireless communication system with one source node and one destination node, which are fixed at two different locations on the ground. Channel condition of the direct link between the source and the destination is undesirable and unable to afford transmission with acceptable performance, e.g., due to obstacles located in line-of-sight area. Thus, a UAV mounted with a transceiver is deployed as an aerial relay to assist communication between the source and the destination. The source, destination, UAV transceiver are all equipped with one antenna. Since most commercial UAVs are energy-constrained due to their limited power capacity of batteries, excessive power consumption for high-performance transmission will undoubtedly drain their batteries and further shorten the operation lifetime. To tackle such an issue, we propose to combine wireless information and power transfer with UAV-assisted cooperative communications, i.e., UAV's wireless transmission capability is exclusively powered by radio signal transmitted from the source while UAV's maneuvering such as taking-off, landing and hovering is still powered by its on-board battery, which can be practically implemented by installing two individual batteries for transceiving and maneuvering, respectively. Such a SWIPT mechanism is able to efficiently lighten the UAV's energy burden by scavenging wireless energy carried by radio signal.

Without loss of generality, we consider a three dimensional Cartesian coordinate system for locations of both source and destination, which are denoted by $[S_{x},S_{y},0]$ and $[D_{x},D_{y},0]$, respectively. We assume that the UAV is flying over the air at a fixed altitude $H$ (e.g., the minimum altitude required for maneuvering without terrain blockage) from a start location to an end location, where supply centers (e.g., charging stations) are located, and serves as a mobile relay for a finite time horizon $T$. For simplicity, the UAV's taking-off and landing are not taken into account and we only focus on the operation period while the UAV is traveling on the aerial plane at altitude $H$. The UAV's instantaneous coordinate is denoted by $[x(t),y(t),H]$ for $0\leq t\leq T$. Specially, the start location $[x_{s},y_{s},H]$ and the end location $[x_{e},y_{e},H]$ are predefined. Hence, the minimum required travel distance from the start location to the end location within time horizon $T$ is given by $d_{min}=\sqrt{(x_{e}-x_{s})^{2}+(y_{e}-y_{s})^{2}}$. The UAV's maximum velocity is denoted by $\overline{V}$ and we assume that $\overline{V}\geq d_{min}/T$ such that the UAV is at least able to find a feasible path from the start location to the end location. 

In general, the UAV's location varies constantly over time while flying along a specific trajectory, which makes the instantaneous cooperative rate between source and destination does as well. However, tedious integral operations have to be involved in precisely computing the end-to-end cooperative throughput while the UAV travels from the start location to the end location. To simplify derivation, we resort to numerical evaluation to approximate the end-to-end cooperative throughput. We discretize the time horizon $T$ into $N$ equally spaced time slots, i.e., $T=N\delta_{t}$, where $\delta_{t}$ denotes duration of each time slot. In practice, $\delta_{t}$ is chosen to be sufficiently small (or equivalently, $N$ is chosen to be sufficiently big) such that we can convincingly assume that the UAV's location stays approximately immobile within each time slot. Hence, the UAV's entire aerial trajectory $(x_{t},y_{t})$ within time horizon $T$ can be approximated with a sequence $[x_{n},y_{n}]$ for $n=1,...,N$, where $[x_{n},y_{n}]$ denotes coordinate of the UAV's location within time slot $n$ and the height coordinate is omitted since the UAV flies at a fixed altitude. Therefore, the UAV's mobility constraints are given by
\begin{equation}
\begin{aligned}
&(x_{1}-x_{s})^{2}+(y_{1}-y_{s})^{2}\leq V^{2}\\
&(x_{n+1}-x_{n})^{2}+(y_{n+1}-y_{n})^{2}\leq V^{2},\quad\forall n=1,...,N-1\\
&(x_{e}-x_{N})^{2}+(y_{e}-y_{N})^{2}\leq V^{2}
\end{aligned}
\end{equation}
where $V\triangleq\overline{V}\delta _{t}$ denotes the UAV's maximum travel distance within each time slot.

We assume that the source-to-UAV and UAV-to-destination channels are dominated by line-of-sight (LOS) links and Doppler effect resulted from the UAV's mobility is able to be compensated perfectly\cite{7572068}. Therefore, within time slot $n$, channel power gain of both source-to-UAV and UAV-to-destination channels follows the free-space path loss model, which are given by
\begin{equation}
h^{sr}_{n}=\dfrac{\beta_{0}}{(d^{sr}_{n})^{2}}=\dfrac{\beta_{0}}{(x_{n}-S_{x})^{2}+(y_{n}-S_{y})^{2}+H^{2}}
\end{equation}
for the source-to-UAV channel and 
\begin{equation}
h^{rd}_{n}=\dfrac{\beta_{0}}{(d^{rd}_{n})^{2}}=\dfrac{\beta_{0}}{(x_{n}-D_{x})^{2}+(y_{n}-D_{y})^{2}+H^{2}}
\end{equation}
for the UAV-to-destination channel, respectively. Here $\beta_{0}$ denotes the channel power gain at the reference distance, e.g., $1$ meter, while $d^{sr}_{n}$ and $d^{rd}_{n}$ denote the source-to-UAV distance and UAV-to-destination distance, respectively. Let $\delta$ and $\delta_{s}$ denote channel noise power and signal processing noise power at the power splitter of the UAV, respectively, and we further define $\gamma_{0}\triangleq\beta_{0}/\delta$ as the reference signal-to-noise ratio (SNR) for source-to-UAV and UAV-to-destination channels and $a\triangleq\delta_{s}/\delta$ as the relative signal processing noise power to simplify derivation. 

\begin{figure}
\centering
 \includegraphics[width=5.0in]{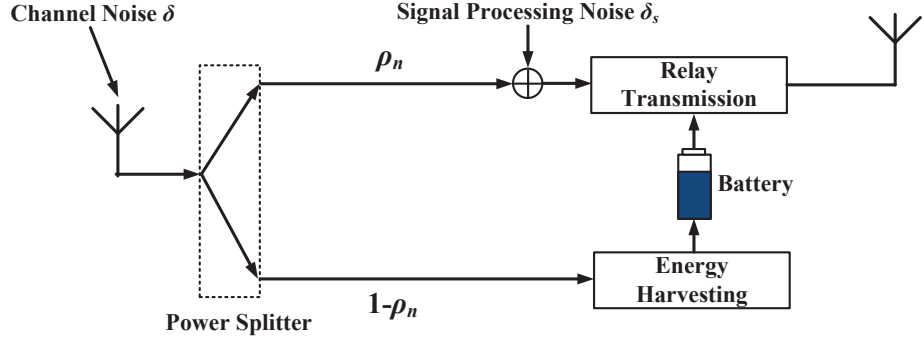}
\caption{Block diagram of the UAV's receiver structure for power-splitting wireless information and power transfer within time slot $n$.}
\label{fig:Power_Split}
\vspace*{-0.15in}
\end{figure}

To serve as a relay, the UAV operates in a time-division half-duplex mode, which is well known due to current limitation in radio implementation. To be more specific, the source transmits data to both the destination and UAV in the first half of each time slot, and then in the second half, the UAV forwards the data to the destination\cite{1362898}. In practice, signal used for relaying modulated information cannot be used for harvesting energy due to hardware limitations\cite{6623062}. Therefore, we consider the power-splitting receiver structure for SWIPT at the UAV, which is one of dominant SWPIT receiver structures. With the power-splitting scheme, an energy harvesting unit and a conventional signal processing core unit are built in the UAV for simultaneous energy harvesting and information relaying. Fig. \ref{fig:Power_Split} shows a typical design of such a structure (a similar block diagram of transceivers for wireless information and power transfer can be found in \cite{6661329}). Within time slot $n$, the received signal at the UAV is split into two power streams in the first half of time slot, which are used for scavenging wireless energy and receiving information from the source with power-splitting ratio $\rho_{n}$ and $1-\rho_{n}$ ($0\leq\rho_{n}\leq 1$), respectively. The harvested energy is used to replenish the UAV's battery for wireless transmission, which is responsible exclusively for information relaying in the second half of the time slot. Hence, within time slot $n$, signal power and noise (including channel noise and signal processing noise) power at the signal processing unit of the UAV are given by $P_{s}\beta_{0}\rho_{n}/(d^{sr}_{n})^{2}$ and $\rho_{n}\delta +\delta_{s}$, respectively. Then SNR at the signal processing unit of the UAV after power splitting can be expressed as $P_{s}\gamma_{0}\rho_{n}/[(\rho_{n}+a)(d^{sr}_{n})^{2}]$. Since not only radio signal from the source but also channel noise accounts for the energy harvested by the UAV, the total energy harvested by the UAV within time slot $n$ amounts to $(P_{s}\beta_{0}/(d^{sr}_{n})^{2}+\delta)(1-\rho_{n})$. Then according to \cite{1362898}, the instantaneous cooperative rate within time slot $n$ is given by
\begin{equation}\label{eqt:af}
r_{n}=\dfrac{1}{2}\log(1+P_{s}\gamma +\dfrac{\dfrac{P_{s}\gamma_{0}^{2}p_{n}\rho_{n}}{(\rho_{n}+a)(d^{sr}_{n}d^{rd}_{n})^{2}}}{1+\dfrac{P_{s}\gamma_{0}\rho_{n}}{(\rho_{n}+a)(d^{sr}_{n})^{2}}+ \dfrac{p_{n}\gamma_{0}}{(d^{rd}_{n})^{2}})}
\end{equation}
for the AF protocol and 
\begin{equation}\label{eqt:df}
r_{n}=\dfrac{1}{2}\min\{\log (1+\dfrac{P_{s}\gamma_{0}\rho_{n}}{(\rho_{n}+a)(d^{sr}_{n})^{2}}), \log (1+P_{s}\gamma+\dfrac{p_{n}\gamma_{0}}{(d^{rd}_{n})^{2}})\}  
\end{equation}
for the DF protocol, respectively, where $P_{s}$ refers to the constant transmission power at the source, $\gamma$ refers to the reference SNR of the source-to-destination channel and $p_{n}$ denotes the UAV's power profile, i.e., relay power within time slot $n$. Within any time slot, the UAV can only use energy that has been harvested in previous time slots for relaying rather than that in future. Therefore, the following energy-causality constraint\cite{5992841} must hold\footnote{We assume that capacity of the UAV's energy storage is sufficient without a overflow.}:
\begin{equation}\label{eqt:energy_causality}
\sum\limits _{i=1}^{n}p_{i}\leq\sum\limits _{i=1}^{n}(1+\dfrac{P_{s}\gamma_{0}}{(d^{sr}_{i})^{2}})\delta(1-\rho_{i}),\quad \forall n=1,...,N
\end{equation}

In this paper, we aim at maximizing the end-to-end cooperative throughput while the UAV travels from the start location to the end location by properly choosing the UAV's power profile $p_{n}$, power-splitting ratio profile $\rho_{n}$ and trajectory $[x_{n},y_{n}]$ for $n=1,...,N$. Such a problem can be formulated as 
\begin{equation}\label{eqt:problem}
\begin{aligned}
\max\limits_{p_{n},\rho_{n},[x_{n},y_{n}]}\quad &\sum\limits _{n=1}^{N}r_{n}\\
{\rm s.t.}\quad &\sum\limits _{i=1}^{n}p_{i}\leq\sum\limits _{i=1}^{n}(1+\dfrac{P_{s}\gamma_{0}}{(d^{sr}_{i})^{2}})\delta(1-\rho_{i}),\quad \forall n=1,...,N\\
&p_{n}\geq 0,\quad 0\leq\rho_{n}\leq 1, \quad \forall n=1,...,N\\
&(x_{1}-x_{s})^{2}+(y_{1}-y_{s})^{2}\leq V^{2}\\
&(x_{n+1}-x_{n})^{2}+(y_{n+1}-y_{n})^{2}\leq V^{2},\quad\forall n=1,...,N-1\\
&(x_{e}-x_{N})^{2}+(y_{e}-y_{N})^{2}\leq V^{2}
\end{aligned}
\end{equation}
Here we assume that transmission on source-to-UAV, UAV-to-destination as well as source-to-destination channels share equal bandwidth. Obviously, (\ref{eqt:problem}) is a non-convex optimization problem, which thus cannot be simply solved with standard convex optimization techniques. In the following sections, we resort to alternate optimization by decomposing (\ref{eqt:problem}) into two subproblems: optimizing power profile and power-splitting ratio profile with fixed trajectory and optimizing trajectory with fixed power profile and power-splitting ratio profile. Based on solutions to the two subproblems, the end-to-end cooperative throughput can be thus iteratively improved until convergence.

\section{Optimizing Power Profile and Power-splitting Ratio Profile}\label{sec:3}
In this section, we focus on the first subproblem of (\ref{eqt:problem}) for both AF and DF protocols, i.e., optimizing the UAV's power profile $p_{n}$ and power-splitting ratio profile  $\rho_{n}$ given its trajectory $[x_{n},y_{n}]$. The analytical results in this section also apply to practical cases, in which the UAV's travel trajectory is predefined for specific tasks instead of being optimized for cooperative transmission performance. Given the UAV's trajectory, (\ref{eqt:problem}) can be reformulated as 
\begin{equation}\label{eqt:problem_profile}
\begin{aligned}
\max\limits_{p_{n},\rho_{n}}\quad &\sum\limits _{n=1}^{N}r_{n}\\
{\rm s.t.}\quad &\sum\limits _{i=1}^{n}p_{i}\leq\sum\limits _{i=1}^{n}(1+\dfrac{P_{s}\gamma_{0}}{(d^{sr}_{i})^{2}})\delta(1-\rho_{i}),\quad \forall n=1,...,N\\
&p_{n}\geq 0,\quad 0\leq\rho_{n}\leq 1, \quad \forall n=1,...,N
\end{aligned}
\end{equation}

\begin{theorem}\label{theorem:concavity}
The objective function in (\ref{eqt:problem_profile}) is concave with respect to $p_{n}$ and $\rho_{n}$ for both AF and DF protocols.
\end{theorem}

\begin{proof}
Please refer to Appendix \ref{app:a} .
\end{proof}

Theorem \ref{theorem:concavity} shows that (\ref{eqt:problem_profile}) is convex for both AF and DF protocols. Therefore, the optimal solution can be derived via the Lagrange dual method.

\subsection{Amplify-and-Forward} 
For AF protocol, it can be verified that (\ref{eqt:problem_profile}) satisfies the Slater’s condition such that strong duality holds and the optimal solution can be derived via solving its dual problem\cite{Boyd:2004:CO:993483}. The Lagrangian of (\ref{eqt:problem_profile}) for AF protocol is given by
\begin{equation}
\begin{aligned}
L_{a}(p_{n},\rho_{n},\lambda_{n})&=\sum\limits_{n=1}^{N}r_{n}+\sum\limits_{n=1}^{N}\lambda_{n}(\sum\limits _{i=1}^{n}(1+\dfrac{P_{s}\gamma_{0}}{(d^{sr}_{i})^{2}})\delta(1-\rho_{i})-p_{i})\\
&=\sum\limits_{n=1}^{N}[r_{n}+\alpha_{n}((1+\dfrac{P_{s}\gamma_{0}}{(d^{sr}_{n})^{2}})\delta(1-\rho_{n})-p_{n})]
\end{aligned}
\end{equation}
where $\lambda_{n}$ refers to the dual variables for the first constraint in (\ref{eqt:problem_profile}) and $\alpha_{n}$ is defined as
\begin{equation}
\alpha_{n}=\sum\limits_{i=n}^{N}\lambda_{n}
\end{equation}
Then the Lagrange dual function of (\ref{eqt:problem_profile}) can be expressed as 
\begin{equation}\label{eqt:problem_profile_dual}
\max\limits_{p_{n}\geq 0,0\leq\rho_{n}\leq 1}\quad g(\lambda_{n})=L_{a}(p_{n},\rho_{n},\lambda_{n})
\end{equation}
and solving (\ref{eqt:problem_profile}) is equivalent to solving its dual problem, which is defined as $\min _{\lambda_{n}\geq 0}g(\lambda _{n})$. In the following, we first solve (\ref{eqt:problem_profile_dual}) with fixed dual variables and then derive the optimal dual variables that minimize the dual problem of (\ref{eqt:problem_profile}).

Given dual variables $\lambda_{n}$, (\ref{eqt:problem_profile_dual}) can be decomposed into $N$ parallel subproblems, which are given by
\begin{equation}\label{eqt:parallel_af}
\max\limits_{p_{n}\geq 0,0\leq\rho_{n}\leq 1}\quad f(p_{n},\rho_{n})=r_{n}+\alpha_{n}[(1+\dfrac{P_{s}\gamma_{0}}{(d^{sr}_{n})^{2}})\delta(1-\rho_{n})-p_{n}]
\end{equation}
for $n=1,...,N$.
According to Theorem \ref{theorem:concavity}, (\ref{eqt:parallel_af}) is convex for each $n$ such that there exists a unique global optimal solution $p^{*}_{n}$ and $\rho^{*}_{n}$, which satisfy the following two conditions:
\begin{equation}\label{eqt:p_af}
p^{*}_{n}=\max\{p^{c}_{n},0\}
\end{equation}
and
\begin{equation}\label{eqt:rho_af}
\rho^{*}_{n}=\left\{
\begin{aligned}
&\rho^{c}_{n},\quad 0\leq\rho^{c}_{n}\leq 1\\
&\mathop{\argmax}_{\rho_{n}\in\{0,1\}}f(p^{*}_{n},\rho_{n}),\quad\text{otherwise}
\end{aligned}
\right.
\end{equation}
where the candidate solution (or stationary point) $p^{c}_{n}$ and $\rho^{c}_{n}$ are defined as
\begin{equation}\label{eqt:candi_p_af}
p^{c}_{n}=\dfrac{b_{n,4}+\sqrt{b_{n,4}^{2}+\dfrac{2b_{n,1}b_{n,2}b_{n,4}}{b_{n,5}}}}{2b_{n,1}b_{n,2}}-\dfrac{b_{n,3}}{b_{n,2}}
\end{equation}
and 
\begin{equation}\label{eqt:candi_rho_af}
\rho^{c}_{n}=\dfrac{c_{n,4}+\sqrt{c_{n,4}^{2}+\dfrac{2c_{n,1}c_{n,2}c_{n,4}}{c_{n,5}}}}{2c_{n,1}c_{n,2}}-\dfrac{c_{n,3}}{c_{n,2}}
\end{equation}
respectively. Detailed proof for (\ref{eqt:p_af}) and (\ref{eqt:rho_af}) as well as definition of $b_{n,i}$ in (\ref{eqt:candi_p_af}) and $c_{n,i}$ in (\ref{eqt:candi_rho_af}) for $i=1,...,5$ are given in Appendix \ref{app:b}. Instead of explicitly solving the equation system composed of (\ref{eqt:p_af}) and (\ref{eqt:rho_af}), one can resort to alternate optimization to obtain the optimal solution, i.e., iteratively update $p_{n}$ following (\ref{eqt:p_af}) with $\rho_{n}$ fixed and update $\rho_{n}$ following (\ref{eqt:rho_af}) with $p_{n}$ fixed. Since each subproblem in (\ref{eqt:parallel_af}) is convex, such a process is bound to make iteratively updated $p_{n}$ and $\rho_{n}$ converge to the global optimal solution.

After obtaining the optimal power profile $p_{n}$ and power-splitting ratio profile $\rho_{n}$ given dual variables $\lambda_{n}$, the dual problem can be efficiently solved via subgradient-based methods, e.g, the ellipsoid method\cite{6760603}, and one feasible subgradient $\bm{d}$ is given by $d_{n}=\sum\limits _{i=1}^{n}(1+\dfrac{P_{s}\gamma_{0}}{(d^{sr}_{i})^{2}})\delta(1-\rho^{*}_{i})-p^{*}_{i}$ for $n=1,...,N$.

\subsection{Decode-and-Forward}
For DF protocol, (\ref{eqt:problem_profile}) can be equivalently formulated as 
\begin{equation}\label{eqt:problem_profile_df}
\begin{aligned}
\max\limits_{p_{n},\rho_{n},\overline{r}_{n}}\quad &\sum\limits _{n=1}^{N}\overline{r}_{n}\\
{\rm s.t.}\quad &\overline{r}_{n}\leq\dfrac{1}{2}\log (1+\dfrac{P_{s}\gamma_{0}\rho_{n}}{(\rho_{n}+a)(d^{sr}_{n})^{2}}),\quad\forall n=1,...,N\\
&\overline{r}_{n}\leq\dfrac{1}{2}\log (1+P_{s}\gamma+\dfrac{p_{n}\gamma_{0}}{(d^{rd}_{n})^{2}}),\quad\forall n=1,...,N\\
&\sum\limits _{i=1}^{n}p_{i}\leq\sum\limits _{i=1}^{n}(1+\dfrac{P_{s}\gamma_{0}}{(d^{sr}_{i})^{2}})\delta(1-\rho_{i}),\quad\forall n=1,...,N\\
&p_{n}\geq 0,\quad 0\leq\rho_{n}\leq 1, \quad \forall n=1,...,N
\end{aligned}
\end{equation}
by introducing variables $\overline{r}_{n}$. Similar with AF protocol, the Lagrangian of (\ref{eqt:problem_profile_df}) is given by
\begin{equation}\label{eqt:Lag_df}
\begin{aligned}
L_{d}(p_{n},\rho_{n},\overline{r}_{n},\lambda_{n},\theta^{1}_{n},\theta^{2}_{n})&=\sum\limits_{n=1}^{N}[\overline{r}_{n}+\alpha_{n}((1+\dfrac{P_{s}\gamma_{0}}{(d^{sr}_{n})^{2}})\delta(1-\rho_{n})-p_{n})\\
&+\theta^{1}_{n}(\dfrac{1}{2}\log (1+\dfrac{P_{s}\gamma_{0}\rho_{n}}{(\rho_{n}+a)(d^{sr}_{n})^{2}})-\overline{r}_{n})+\theta^{2}_{n}(\dfrac{1}{2}\log (1+P_{s}\gamma+\dfrac{p_{n}\gamma_{0}}{(d^{rd}_{n})^{2}})-\overline{r}_{n})]
\end{aligned}
\end{equation}
where $\theta^{1}_{n}$ and $\theta^{2}_{n}$ refer to the dual variables for the first two constraints in (\ref{eqt:problem_profile_df}). Then the Lagrange dual function of (\ref{eqt:problem_profile_df}) can be expressed as 
\begin{equation}\label{eqt:problem_profile_dual2}
\max\limits_{p_{n}\geq 0,0\leq\rho_{n}\leq 1,\overline{r}_{n}}\quad g(\lambda_{n},\theta^{1}_{n},\theta^{2}_{n})=L_{d}(p_{n},\rho_{n},\overline{r}_{n},\lambda_{n},\theta^{1}_{n},\theta^{2}_{n})
\end{equation}
and solving (\ref{eqt:problem_profile_df}) is equivalent to solving its dual problem defined as $\min _{\lambda_{n},\theta^{1}_{n},\theta^{2}_{n}\geq 0}g(\lambda _{n},\theta^{1}_{n},\theta^{2}_{n})$. In the following, we first solve (\ref{eqt:problem_profile_dual2}) with fixed dual variables and then derive the optimal dual variables that minimize the dual problem.

It is easy to find that $\theta^{1}_{n}+\theta^{2}_{n}=1$ must hold for $n=1,...,N$, otherwise it leads to unbounded end-to-end cooperative throughput (which can also be inferred from the Karush-Kuhn-Tucker optimality conditions). Thus, (\ref{eqt:Lag_df}) can be rewritten as
\begin{equation}\label{eqt:Lag_df2}
\begin{aligned}
L_{d}(p_{n},\rho_{n},\lambda_{n},\theta_{n})&=\sum\limits_{n=1}^{N}[\alpha_{n}((1+\dfrac{P_{s}\gamma_{0}}{(d^{sr}_{n})^{2}})\delta(1-\rho_{n})-p_{n})\\
&+\dfrac{\theta_{n}}{2}\log (1+\dfrac{P_{s}\gamma_{0}\rho_{n}}{(\rho_{n}+a)(d^{sr}_{n})^{2}})+\dfrac{1-\theta_{n}}{2}\log (1+P_{s}\gamma+\dfrac{p_{n}\gamma_{0}}{(d^{rd}_{n})^{2}})]
\end{aligned}
\end{equation}
where $\theta_{n}=\theta^{1}_{n}=1-\theta^{2}_{n}$. Then given dual variables $\lambda_{n}$ and $\theta_{n}$, (\ref{eqt:problem_profile_df}) can be decomposed into $2N$ parallel subproblems as well, which are given by
\begin{equation}\label{eqt:parallel_df}
\left.\{
\begin{aligned}
&\max\limits_{p_{n}\geq 0}\quad f_{1}(p_{n})=\dfrac{1-\theta_{n}}{2}\log (1+P_{s}\gamma+\dfrac{p_{n}\gamma_{0}}{(d^{rd}_{n})^{2}})-\alpha_{n}p_{n}\\
&\max\limits_{0\leq\rho_{n}\leq 1}\quad f_{2}(\rho_{n})=\alpha_{n}(1+\dfrac{P_{s}\gamma_{0}}{(d^{sr}_{n})^{2}})\delta(1-\rho_{n})+\dfrac{\theta_{n}}{2}\log (1+\dfrac{P_{s}\gamma_{0}\rho_{n}}{(\rho_{n}+a)(d^{sr}_{n})^{2}})
\end{aligned}
\right.
\end{equation}
for $n=1,...,N$. Obviously, the two objective functions in (\ref{eqt:parallel_df}) are concave such that the optimal $p^{*}_{n}$ and $\rho^{*}_{n}$ can be derived as 
\begin{equation}\label{eqt:p_df}
p^{*}_{n}=\max\{p^{c}_{n},0\}
\end{equation}
and
\begin{equation}\label{eqt:rho_df}
\rho^{*}_{n}=\left.\{
\begin{aligned}
&\rho^{c}_{n},\quad 0\leq\rho^{c}_{n}\leq 1\\
&\mathop{\argmax}_{\rho_{n}\in\{0,1\}}f(p^{*}_{n},\rho_{n}),\quad\text{otherwise}
\end{aligned}
\right.
\end{equation}
where the candidate solution (or stationary point) $p^{c}_{n}$ and $\rho^{c}_{n}$ are defined as
\begin{equation}\label{eqt:candi_p_df}
p^{c}_{n}=\dfrac{1-\theta_{n}}{2\alpha_{n}}-\dfrac{(1+P_{s}\gamma)(d^{rd}_{n})^{2}}{\gamma_{0}}
\end{equation}
and 
\begin{equation}\label{eqt:candi_rho_df}
\rho^{c}_{n}=\dfrac{\dfrac{\alpha_{n}\delta aP_{s}\gamma_{0}}{\theta_{n}(d^{sr}_{n})^{2}}    
+\sqrt{(\dfrac{\alpha_{n}\delta aP_{s}\gamma_{0}}{\theta_{n}(d^{sr}_{n})^{2}})^2+\dfrac{2\alpha_{n}\delta aP_{s}\gamma_{0}}{\theta_{n}(d^{sr}_{n})^{2}}}}{\dfrac{2\alpha_{n}\delta}{\theta_{n}}(1+\dfrac{P_{s}\gamma_{0}}{(d^{sr}_{n})^{2}})}-a
\end{equation}
respectively. Please refer to Appendix \ref{app:c} for proof of (\ref{eqt:p_df}) and (\ref{eqt:rho_df}).

Similar with the AF protocol, the dual problem can be solved via subgradient-based methods as well and the  subgradient $\bm{d}$ can be defined as 
\begin{equation}
d_{n}=\left\{
\begin{aligned}
&\sum\limits _{i=1}^{n}(1+\dfrac{P_{s}\gamma_{0}}{(d^{sr}_{i})^{2}})\delta(1-\rho^{*}_{i})-p^{*}_{i},\quad n=1,...,N\\
&\dfrac{1}{2}[\log (1+\dfrac{P_{s}\gamma_{0}\rho^{*}_{n}}{(\rho^{*}_{n}+a)(d^{sr}_{n})^{2}})-\log (1+P_{s}\gamma+\dfrac{p^{*}_{n}\gamma_{0}}{(d^{rd}_{n})^{2}})],\quad n=N+1,...,2N
\end{aligned}
\right.
\end{equation}

Thus, the algorithm for optimizing the UAV's power profile and power-splitting ratio profile is summarized in \ref{alg:profile}.

\begin{algorithm}[htb] 
\caption{Profile Optimization with Fixed Trajectory}\label{alg:profile}
\begin{algorithmic}[1]
\STATE Initialize the dual variables $\lambda_{n}\geq 0$ and $0\leq\theta_{n}\leq 1$ (for the DF protocol only) for $n=1,...,N$;
\STATE \textbf{Repeat}
\STATE \quad Derive the optimal power profile and power-splitting ratio profile following (\ref{eqt:p_af}) and (\ref{eqt:rho_af}) for the AF protocol, or (\ref{eqt:p_df}) and (\ref{eqt:rho_df}) for the DF protocol;
\STATE \quad Compute the subgradient of the Lagrange dual function;
\STATE \quad Obtain the optimal dual variables  $\lambda^{*}_{n}$ and $\theta^{*}_{n}$ via the ellipsoid method;
\STATE \textbf{Until} Improvement of the objective function value or variation of the optimal dual variables  $\lambda^{*}_{n}$ and $\theta^{*}_{n}$ converges under a predefined tolerance;
\STATE Output $p^{*}_{n}$ and $\rho^{*}_{n}$ along with $\lambda^{*}_{n}$ and $\theta^{*}_{n}$ as the optimal solution to the UAV's power profile and power-splitting ratio profile.
\end{algorithmic}
\end{algorithm}

\section{Optimizing Trajectory}\label{sec:4}
In this section, we focus on the other subproblem of (\ref{eqt:problem}) for both AF and DF protocols, i.e., optimizing the UAV's trajectory $[x_{n},y_{n}]$ given its power profile $p_{n}$ and power-split ratio profile $\rho_{n}$. The analytical results in this section also apply to cases where the UAV operates with predefined transmission power and power-split ratio, e.g., the greedy strategy in \cite{7417077}. Given the UAV's power profile and power-split ratio profile, (\ref{eqt:problem}) can be reformulated as 
\begin{equation}\label{eqt:trajectory}
\begin{aligned}
\max\limits_{[x_{n},y_{n}]}\quad &\sum\limits _{n=1}^{N}r_{n}\\
{\rm s.t.}\quad &(x_{1}-x_{s})^{2}+(y_{1}-y_{s})^{2}\leq V^{2}\\
&(x_{n+1}-x_{n})^{2}+(y_{n+1}-y_{n})^{2}\leq V^{2},\quad\forall n=1,...,N-1\\
&(x_{e}-x_{N})^{2}+(y_{e}-y_{N})^{2}\leq V^{2}
\end{aligned}
\end{equation}
Obviously, (\ref{eqt:trajectory}) is non-convex since the objective function is non-concave over trajectory $[x_{n},y_{n}]$ for both AF and DF protocols. Therefore, we resort to the successive convex optimization technique in \cite{7572068}, which iteratively maximizes a lower bound of (\ref{eqt:trajectory}) by optimizing an incremental trajectory. We specify that after $l$ rounds of iteration, the UAV's trajectory and the cooperative rate are denoted by $[x_{n,l},y_{n,l}]$ and $r_{n,l}$, respectively. Let $[\Delta x_{n,l},\Delta y_{n,l}]$ denote the incremental trajectory, i.e., we have $x_{n,l+1}=x_{n,l}+\Delta x_{n,l}$ and $y_{n,l+1}=y_{n,l}+\Delta y_{n,l}$ for trajectory update  after $l$ rounds of iteration.

\subsection{Amplify-and-Forward} 
To facilitate the following derivation, we first define function $f(z_{1},z_{2})$ as
\begin{equation}\label{eqt:function_z}
f(z_{1},z_{2})=\log(\omega +\dfrac{\dfrac{r_{1}r_{2}}{(A_{1}+z_{1})(A_{2}+z_{2})}}{1+\dfrac{r_{1}}{A_{1}+z_{1}}+\dfrac{r_{2}}{A_{2}+z_{2}}})
\end{equation}
where $\omega >0$, $z_{i}>-A_{i}$ and $r_{i}>0$ for $i\in\{1,2\}$.

\begin{lemma}\label{lemma:convexity}
The function defined in (\ref{eqt:function_z}) is convex with respect to $z_{1}$ and $z_{2}$.
\end{lemma}

\begin{proof}
Please refer to Appendix \ref{app:d} .
\end{proof}

\begin{theorem}\label{theorem:lb_af}
Given current trajectory $[x_{n,l},y_{n,l}]$, the following inequality holds:
\begin{equation}\label{eqt:lb_af}
r_{n,l+1}\geq r^{\prime}_{n,l+1}\triangleq r_{n,l}-\mu_{n,l}(\Delta x_{n,l}^{2}+\Delta y_{n,l}^{2})-\delta_{n,l}\Delta x_{n,l}-\eta_{n,l}\Delta y_{n,l}
\end{equation}
for any incremental trajectory $[\Delta x_{n,l},\Delta y_{n,l}]$, where coefficients $\mu_{n,l}$, $\delta_{n,l}$ and $\eta_{n,l}$ are defined in (\ref{eqt:def}). 
\end{theorem}

\begin{proof}
Please refer to Appendix \ref{app:e} .
\end{proof}

Theorem \ref{theorem:lb_af} indicates that given the UAV's current trajectory $[x_{n,l},y_{n,l}]$ and an incremental trajectory $[\Delta x_{n,l},\Delta y_{n,l}]$ after $l$ rounds of iteration, the resulting cooperative rate $r_{n,l+1}$ in the $l+1$th iteration is lower-bounded by $r^{\prime}_{n,l+1}$, which is a concave quadratic functions with respect to the incremental trajectory $[\Delta x_{n,l},\Delta y_{n,l}]$ since $\mu_{n,l}>0$ (which can be inferred from $\phi_{n,l}>0$). Then given the UAV's current trajectory $[x_{n,l},y_{n,l}]$, the optimum of (\ref{eqt:trajectory}) for the AF protocol is lower-bounded by that of the following problem
\begin{equation}\label{eqt:problem_lb_af}
\begin{aligned}
\max\limits_{[\Delta x_{n,l},\Delta y_{n,l}]}\quad &\sum\limits _{n=1}^{N}r^{\prime}_{n,l+1}\\
{\rm s.t.}\quad &(x_{1,l}+\Delta x_{1,l}-x_{s})^{2}+(y_{1,l}+\Delta y_{1,l}-y_{s})^{2}\leq V^{2}\\
&(x_{n+1,l}+\Delta x_{n+1,l}-x_{n,l}-\Delta x_{n,l})^{2}+\\
&(y_{n+1,l}+\Delta y_{n+1,l}-y_{n,l}-\Delta y_{n,l})^{2}\leq V^{2},\quad\forall n=1,...,N-1\\
&(x_{e}-x_{N,l}-\Delta x_{N,l})^{2}+(y_{e}-y_{N,l}-\Delta y_{N,l})^{2}\leq V^{2}
\end{aligned}
\end{equation}

Obviously, (\ref{eqt:problem_lb_af}) is a convex problem (quadratically constrained quadratic programming) with respect to the incremental trajectory $[\Delta x_{n,l},\Delta y_{n,l}]$ due to the concave objective function and convex-set constraints. Therefore, (\ref{eqt:problem_lb_af}) can be readily solved via existing optimization techniques (e.g., the interior point method). Then following the successive convex optimization technique in \cite{7572068}, for the AF protocol, (\ref{eqt:trajectory}) can be approximately solved by iteratively solving (\ref{eqt:problem_lb_af}) for the optimal incremental trajectory and updating the UAV's trajectory.

\subsection{Decode-and-Forward}
Similarly, for the DF protocol, one can still leverage the successive convex optimization technique in \cite{7572068}. We first define $r^{1}_{n,l}$ and $r^{2}_{n,l}$ as the two terms in (\ref{eqt:df}) to facilitate derivation. According to Lemma 2 in \cite{7572068}, given the UAV's trajectory $[x_{n,l},y_{n,l}]$ and an incremental trajectory $[\Delta x_{n,l},\Delta y_{n,l}]$ after $l$ rounds of iteration, $r^{1}_{n,l+1}$ and $r^{2}_{n,l+1}$ are lower-bounded by
\begin{equation}\label{eqt:lb_df}
r^{i}_{n,l+1}\geq r^{i\prime}_{n,l+1}\triangleq r^{i}_{n,l}-\mu^{i}_{n,l}(\Delta x^{2}_{n,l}+\Delta y^{2}_{n,l})-\delta^{i}_{n,l}\Delta x_{n,l}-\eta^{i}_{n,l}\Delta y_{n,l}
\end{equation}
for $i\in\{1,2\}$, where $\mu^{i}_{n,l}$, $\delta^{i}_{n,l}$ and $\eta^{i}_{n,l}$ are defined as
\begin{equation}
\left.\{
\begin{aligned}
&\mu^{1}_{n,l}=\dfrac{\dfrac{P_{s}\gamma_{0}\rho_{n}}{\rho_{n}+a}}{(d^{sr}_{n,l})^{2}[\dfrac{P_{s}\gamma_{0}\rho_{n}}{\rho_{n}+a}+(d^{sr}_{n,l})^{2}]}\\
&\mu^{2}_{n,l}=\dfrac{p_{n}\gamma_{0}}{(d^{rd}_{n,l})^{2}[p_{n}\gamma_{0}+(1+P_{s}\gamma)(d^{rd}_{n,l})^{2}]}\\
&\delta^{1}_{n,l}=2\mu^{1}_{n,l}(x_{n,l}-S_{x}),\quad \delta^{2}_{n,l}=2\mu^{2}_{n,l}(x_{n,l}-D_{x})\\
&\eta^{1}_{n,l}=2\mu^{1}_{n,l}(y_{n,l}-S_{y}),\quad \eta^{2}_{n,l}=2\mu^{2}_{n,l}(y_{n,l}-D_{y})
\end{aligned}
\right.
\end{equation}
Then given the UAV's current trajectory $[x_{n,l},y_{n,l}]$, the optimum of (\ref{eqt:trajectory}) for the DF protocol is lower-bounded by that of the following problem
\begin{equation}\label{eqt:problem_lb_df}
\begin{aligned}
\max\limits_{[\Delta x_{n,l},\Delta y_{n,l}]}\quad &\sum\limits _{n=1}^{N}\min\{r^{1\prime}_{n,l+1},r^{2\prime}_{n,l+1}\}\\
&(x_{n+1,l}+\Delta x_{n+1,l}-x_{n,l}-\Delta x_{n,l})^{2}+\\
&(y_{n+1,l}+\Delta y_{n+1,l}-y_{n,l}-\Delta y_{n,l})^{2}\leq V^{2},\quad\forall n=1,...,N-1\\
&(x_{e}-x_{N,l}-\Delta x_{N,l})^{2}+(y_{e}-y_{N,l}-\Delta y_{N,l})^{2}\leq V^{2}
\end{aligned}
\end{equation}

Similar with (\ref{eqt:problem_lb_af}), (\ref{eqt:problem_lb_df}) is a convex problem since both $r^{1\prime}_{n,l+1}$ and $r^{2\prime}_{n,l+1}$ are concave with respect to the incremental trajectory $[\Delta x_{n,l},\Delta y_{n,l}]$ according to (\ref{eqt:lb_df}). Therefore, for the DF protocol, (\ref{eqt:trajectory}) can be approximately solved by iteratively solving (\ref{eqt:problem_lb_df}) for the optimal incremental trajectory and updating the UAV's trajectory as well.

Thus, the algorithm for optimizing the UAV's trajectory with fixed power profile and power-splitting ratio profile is summarized in \ref{alg:trajectory}.

\begin{algorithm}[htb] 
\caption{Trajectory Optimization with Fixed Profile}\label{alg:trajectory}
\begin{algorithmic}[1]
\STATE Initialize the the UAV's trajectory $[x_{n,l},y_{n,l}]$ for $n=1,...,N$ that satisfies the constraints in (\ref{eqt:trajectory}) and set $l=0$;
\STATE \textbf{Repeat}
\STATE \quad Solve (\ref{eqt:problem_lb_af}) for the AF protocol or (\ref{eqt:problem_lb_df}) for the DF protocol via the interior point method for the optimal incremental trajectory $[\Delta^{*}x_{n,l},\Delta^{*}y_{n,l}]$;
\STATE \quad Update the UAV's trajectory by $x_{n,l+1}=x_{n,l}+\Delta^{*}x_{n,l}$ and $y_{n,l+1}=y_{n,l}+\Delta^{*}y_{n,l}$ for $n=1,...,N$;
\STATE \quad Set $l=l+1$;
\STATE \textbf{Until} Improvement of the lower bound or variation of the optimal incremental trajectory  $[\Delta^{*}x_{n,l},\Delta^{*}y_{n,l}]$ converges under a predefined tolerance;
\STATE Output $[x_{n,l},y_{n,l}]$ as the UAV's optimal trajectory.
\end{algorithmic}
\end{algorithm}

\section{Alternate Optimization}
Based on the solutions to the two subproblems of (\ref{eqt:problem}) as discussed in Section \ref{sec:3} and \ref{sec:4}, an iterative algorithm based on alternate optimization is proposed to jointly optimize the UAV's power profile, power-splitting ratio profile and trajectory, as summarized in Algorithm \ref{alg:alternate_optimize}.

Apparently, Algorithm \ref{alg:alternate_optimize} involves only convex problem solving in each iteration according to Algorithm \ref{alg:profile} and \ref{alg:trajectory}. Therefore, the overall complexity of Algorithm \ref{alg:alternate_optimize} is polynomial for the worst case because number of iterations is only related to stop criterion instead of problem scale $N$. Moreover, it is also notable that since Algorithm \ref{alg:alternate_optimize} is based on alternate optimization, global optimum cannot be guaranteed in theory. Hence, the optimal cooperative throughput achieved by Algorithm \ref{alg:alternate_optimize} might be affected by the UAV's initial trajectory.

\begin{algorithm}[htb] 
\caption{Alternate Optimization}\label{alg:alternate_optimize}
\begin{algorithmic}[1]
\STATE Initialize the UAV's trajectory;
\STATE \textbf{Repeat}
\STATE \quad Optimize the UAV's power profile and power-splitting ratio profile with fixed trajectory following Algorithm \ref{alg:profile};
\STATE \quad Optimize the UAV's trajectory with fixed power profile and power-splitting ratio profile following Algorithm \ref{alg:trajectory};
\STATE \textbf{Until} Improvement of the end-to-end cooperative throughput converges under a predefined tolerance;
\STATE Output the UAV's power profile, power-splitting ratio and trajectory as the optimal solution to (\ref{eqt:problem}).
\end{algorithmic}
\end{algorithm}

\section{Numerical Results}
In this section, numerical results are presented to validate the proposed optimal design for the UAV-assisted cooperative communication system.  For system layout, source and destination are located at $[0,0]$ and $[2,0]$, respectively, and the UAV flies from $[0,1]$ to $[2,-1]$ at altitude $H=0.3$ with maximum travel distance per time slot $V=0.2$. For other system parameters, we have $\gamma_{0}=1$, $\gamma=0.01$, $P_{s}=1$, $\delta =1$ and $a=2$. Besides, to balance accuracy and computational complexity, we focus on a time horizon of $50$ time slots. One of our previous works is involved for comparison as ``greedy'' strategy since in each time slot, the UAV uses up all the energy harvested from signal sent from the source for relaying\cite{7417077}. In other words, with the greedy strategy, the energy-causality constraint (\ref{eqt:energy_causality}) is transformed into
\begin{equation}\label{eqt:greedy}
p_{n}=(1+\dfrac{P_{s}\gamma_{0}}{(d^{sr}_{n})^{2}})\delta(1-\rho_{n}),\quad \forall n=1,...,N
\end{equation}
Therefore, variables $p_{n}$ are eliminated and the optimal power-splitting ratio profile given the UAV's trajectory can be derived by following (17) and (23) in \cite{7417077} for AF and DF protocols.

\subsection{Profile Optimization Given Trajectory}

\begin{figure}
\centering
 \includegraphics[width=4.3in]{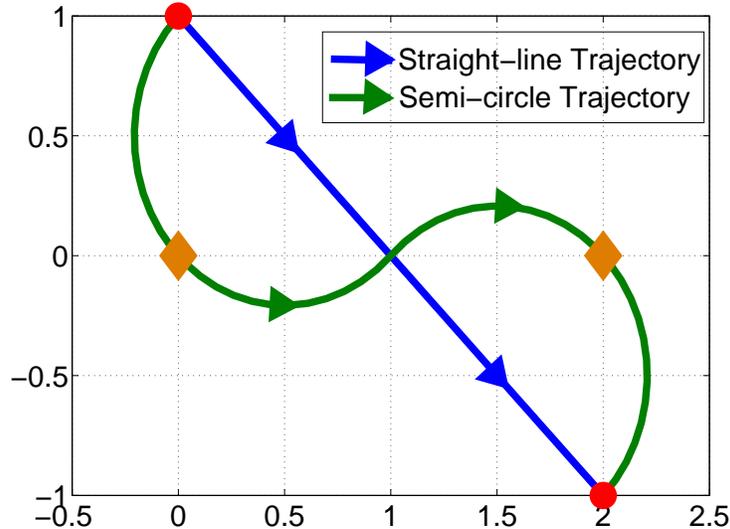}
\caption{Two specific trajectories for the UAV's optimal power profile and power-splitting profile. The yellow diamonds are the source and destination and the red circles are the UAV's start and end locations.}
\label{fig:trajectory}
\vspace*{-0.15in}
\end{figure}

We first focus on the UAV's optimal power profile and power-splitting ratio profile with two specific trajectories from source to destination: straight-line trajectory and semi-circle trajectory, as shown in Fig. \ref{fig:trajectory}. The UAV's optimal power profile with the two trajectories is shown in Fig. \ref{fig:PowerProfileStraight} and \ref{fig:PowerProfileCircle}, respectively. It can be seen that the greedy strategy always exhausts all the energy harvested from the source in each time slot, while the optimal power profile is more long-sighted and makes use of the harvested energy more wisely. For instance, with the AF protocol, the UAV tends to save energy in the first few time slots and consumes more in the future. It is notable that with the longer semi-circle trajectory, the UAV is able to harvest more energy from the source.

The UAV's optimal power-splitting ratio profile with the two trajectories are presented in Fig. \ref{fig:PSRatioProfileStraight} and \ref{fig:PSRatioProfileCircle}, respectively. It can be found that the optimal power-splitting ratio profile is significantly different from that with the greedy strategy. Similarly, it demonstrates that the UAV tends to harvest energy without relaying at the beginning. Besides, by comparing the results with the two trajectories, we find that the power-splitting ratio with the semi-circle trajectory is in general higher than that with the straight-line trajectory. This is probably because flying along the semi-circle trajectory leads to better channel condition due to higher chance of getting closer to either source or destination, as shown in Fig. \ref{fig:trajectory}.	

It is also noticeable in Fig. \ref{fig:PowerProfileStraight}-\ref{fig:PSRatioProfileCircle} that the greedy strategy could be identical with the optimal solution. Take the AF protocol with straight-line trajectory for example (as shown in Fig. \ref{fig:PowerProfileStraight} and \ref{fig:PSRatioProfileStraight}), the optimal solution converges to the greedy strategy after the $34$th time slot, which means the optimal choice is to myopically use up all the harvested energy without saving in each timeslot. In this sense, the greedy strategy is viable for a certain setting of system parameters. Practically, one can compare (\ref{eqt:greedy}) with (\ref{eqt:p_af}) for the AF protocol or with (\ref{eqt:p_df}) for the DF protocol to see whether the greedy strategy is optimal.

\def\2Win{4.3in}
\begin{figure*}
\centering
\begin{minipage}[t]{\2Win}
\includegraphics[width=\2Win]{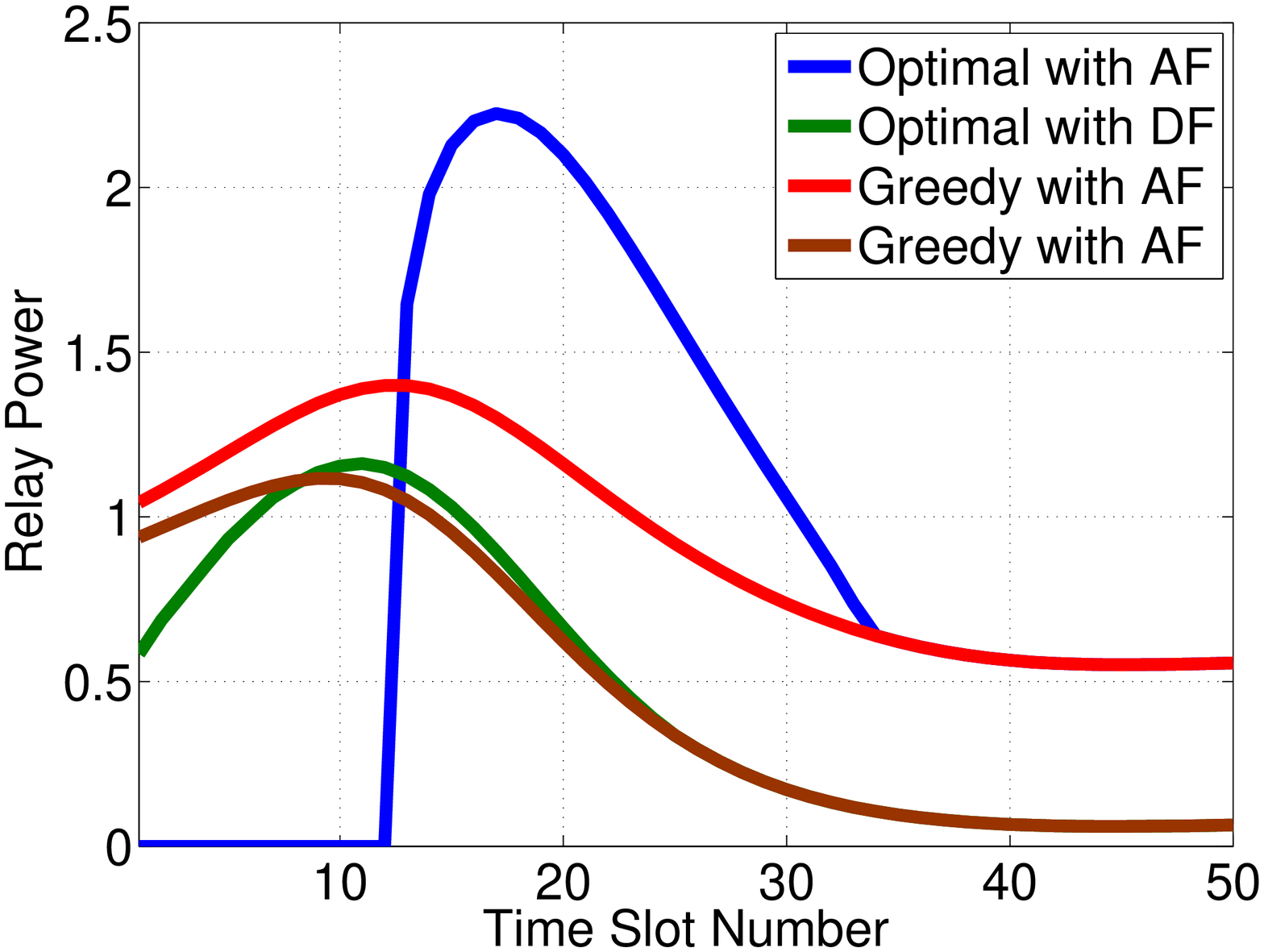}
\caption{The UAV's optimal power profile with straight-line trajectory.}
\label{fig:PowerProfileStraight}
\end{minipage}
\begin{minipage}[t]{\2Win}
\includegraphics[width=\2Win]{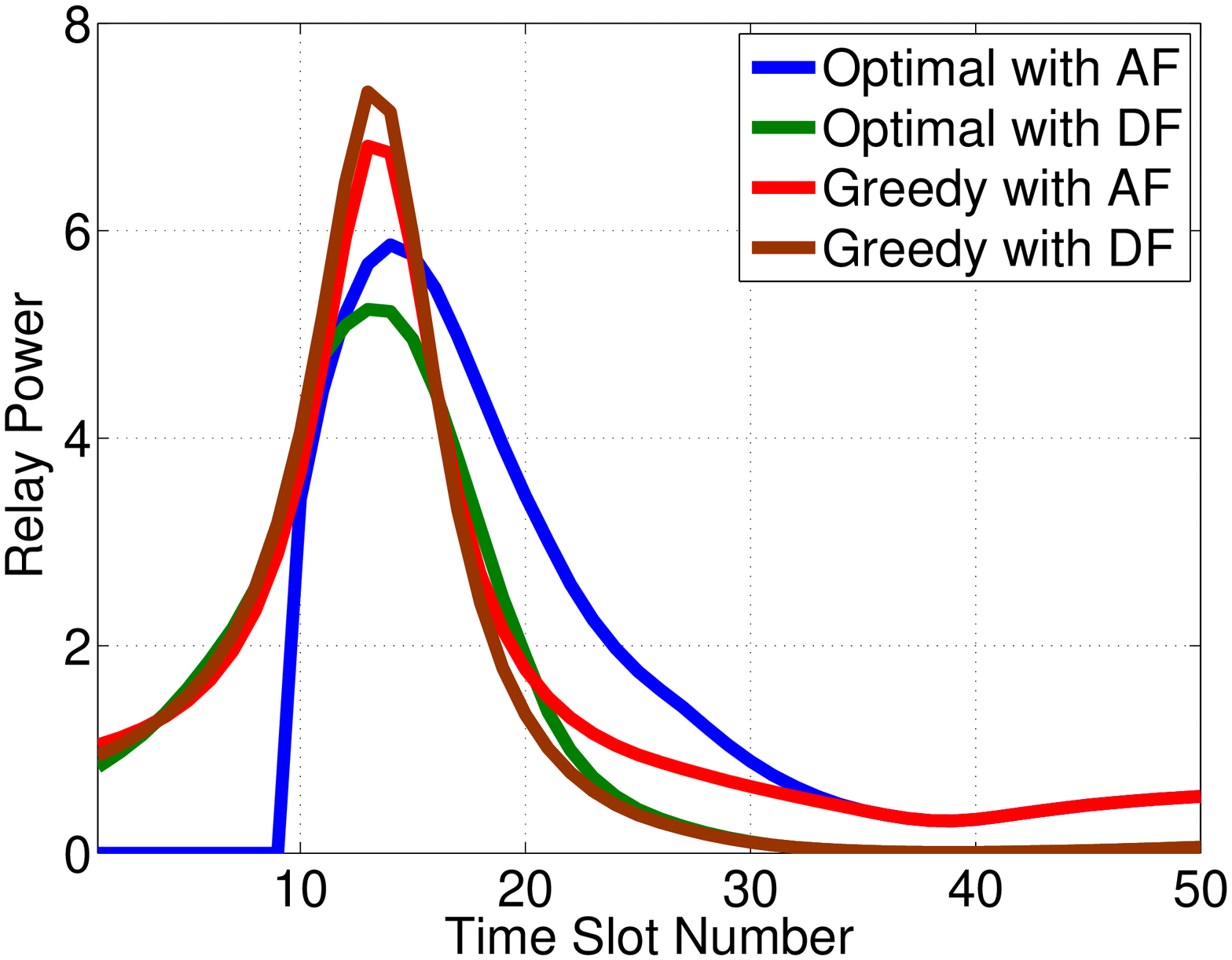}
\caption{The UAV's optimal power profile with semi-circle trajectory.}
\label{fig:PowerProfileCircle}
\end{minipage}
\end{figure*}

\def\2Win{4.3in}
\begin{figure*}
\centering
\begin{minipage}[t]{\2Win}
\includegraphics[width=\2Win]{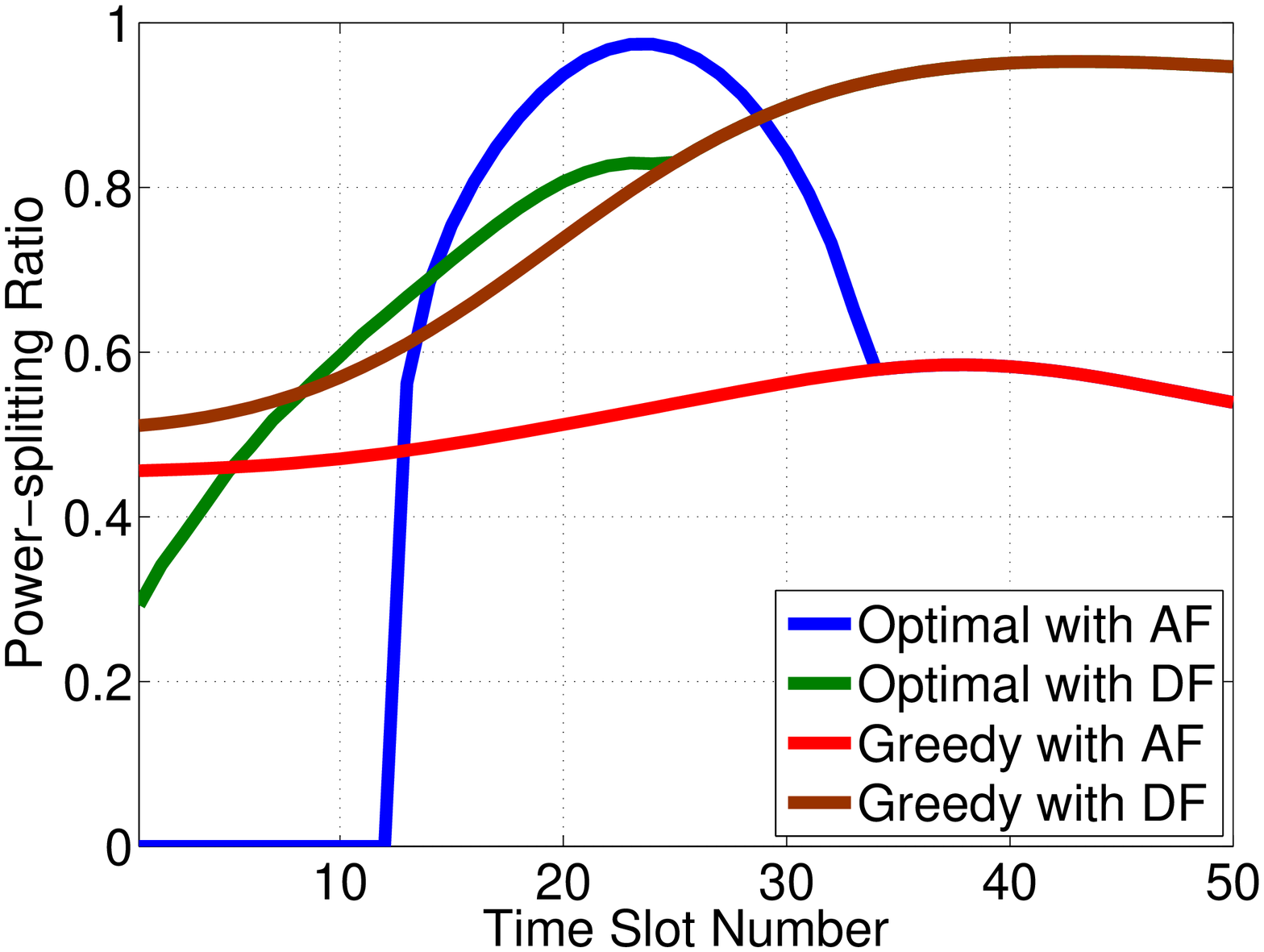}
\caption{The UAV's optimal power-splitting ratio profile with straight-line trajectory.}
\label{fig:PSRatioProfileStraight}
\end{minipage}
\begin{minipage}[t]{\2Win}
\includegraphics[width=\2Win]{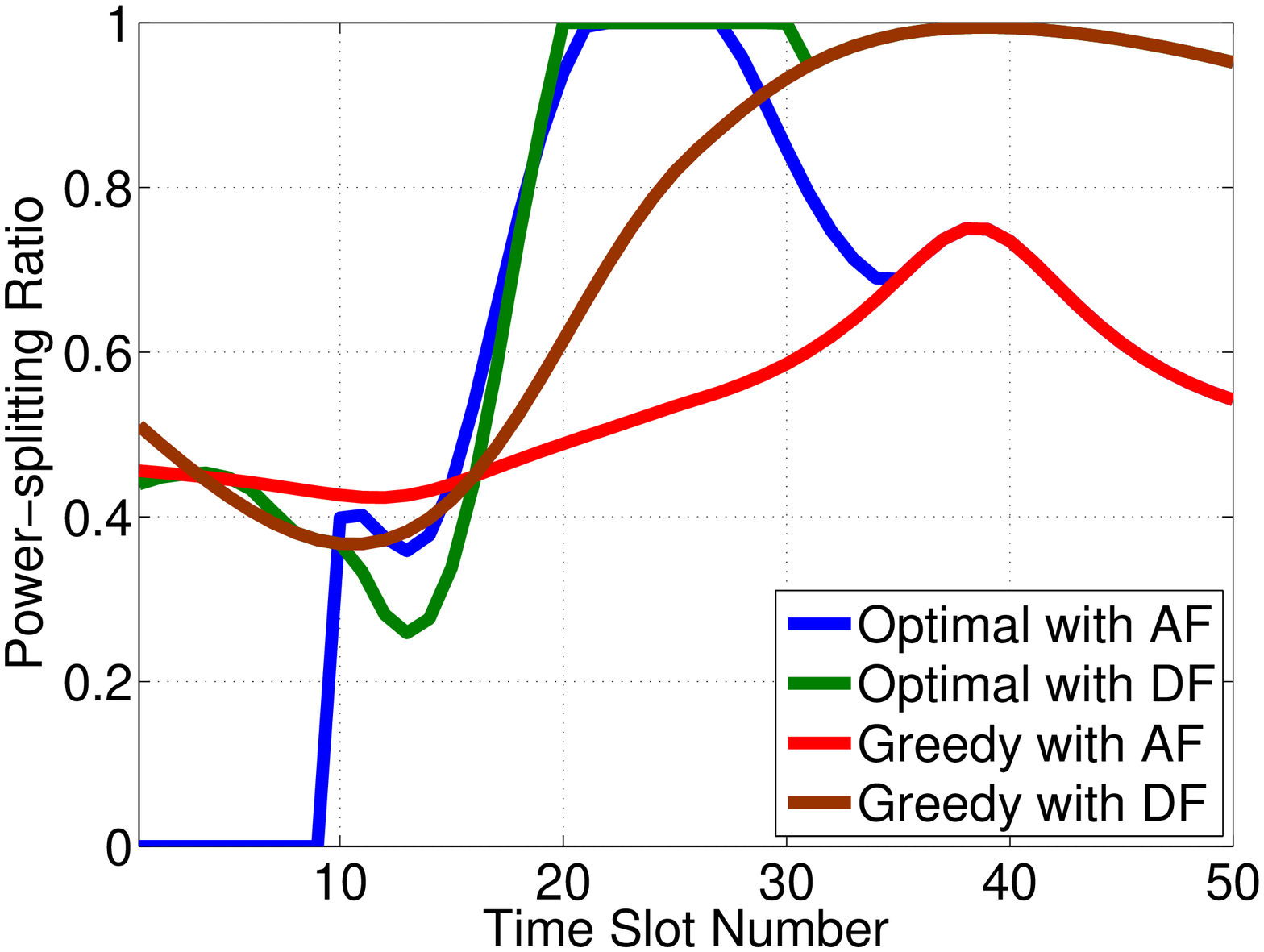}
\caption{The UAV's optimal power-splitting ratio profile with semi-circle trajectory.}
\label{fig:PSRatioProfileCircle}
\end{minipage}
\end{figure*}

\subsection{Trajectory Optimization Given Profile}
In this subsection, we focus on the UAV's optimal trajectory with fixed power profile and power-splitting ratio profile. For demonstration purpose, the greedy strategy with constant power profile and power-splitting ratio profile through the entire $N$ time slots is considered, i.e., power-splitting ratio is fixed (at $0.5$ throughout the experiments in this subsection) and power profile is accordingly determined by (\ref{eqt:greedy}). Besides, we take the straight-line trajectory in Fig. \ref{fig:trajectory} as the initial trajectory for Algorithm \ref{alg:trajectory}.

\def\2Win{4.3in}
\begin{figure*}
\centering
\begin{minipage}[t]{\2Win}
\includegraphics[width=\2Win]{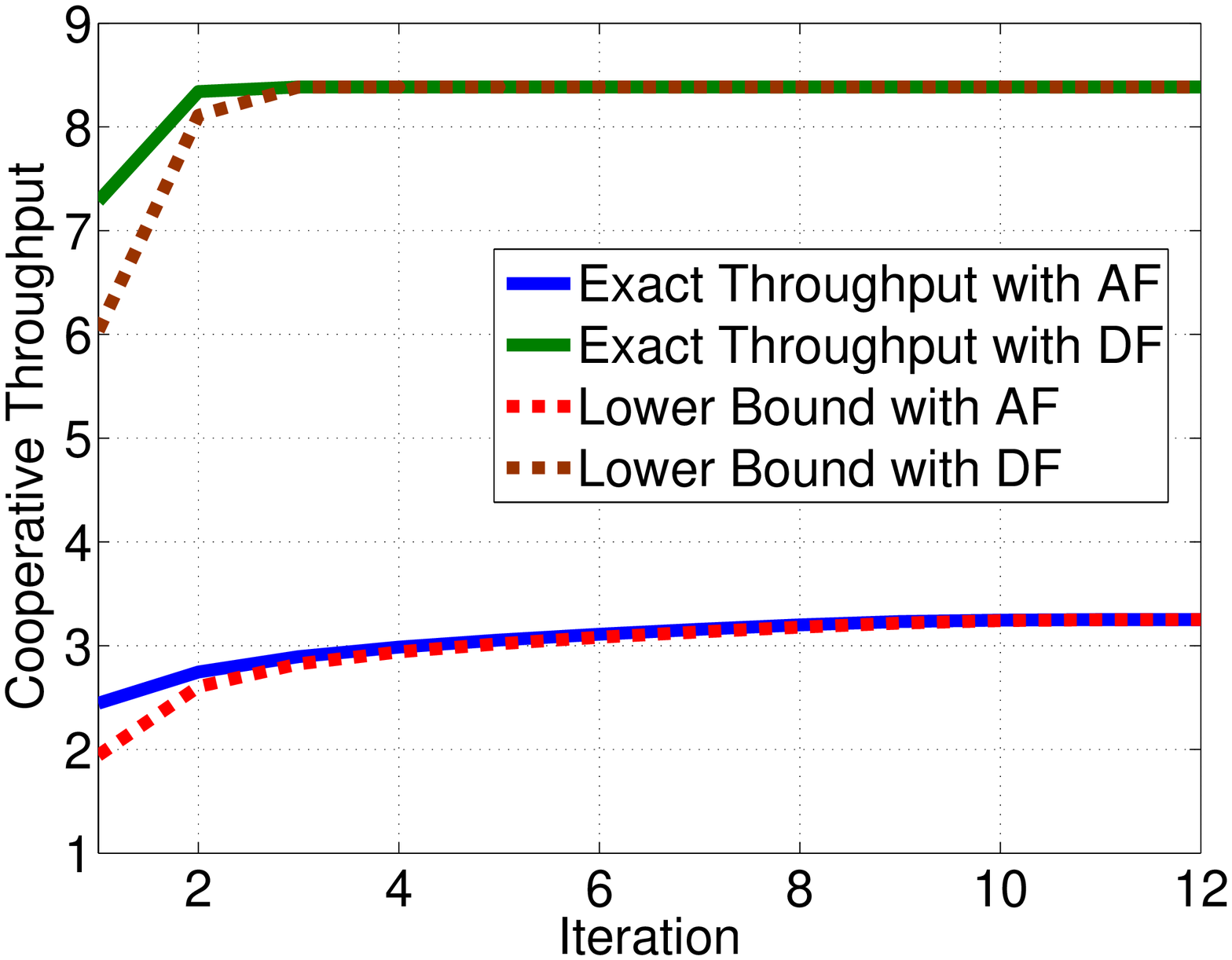}
\caption{Cooperative throughput iteration by Algorithm \ref{alg:trajectory}.}
\label{fig:ThroughputIteration}
\end{minipage}
\begin{minipage}[t]{\2Win}
\includegraphics[width=\2Win]{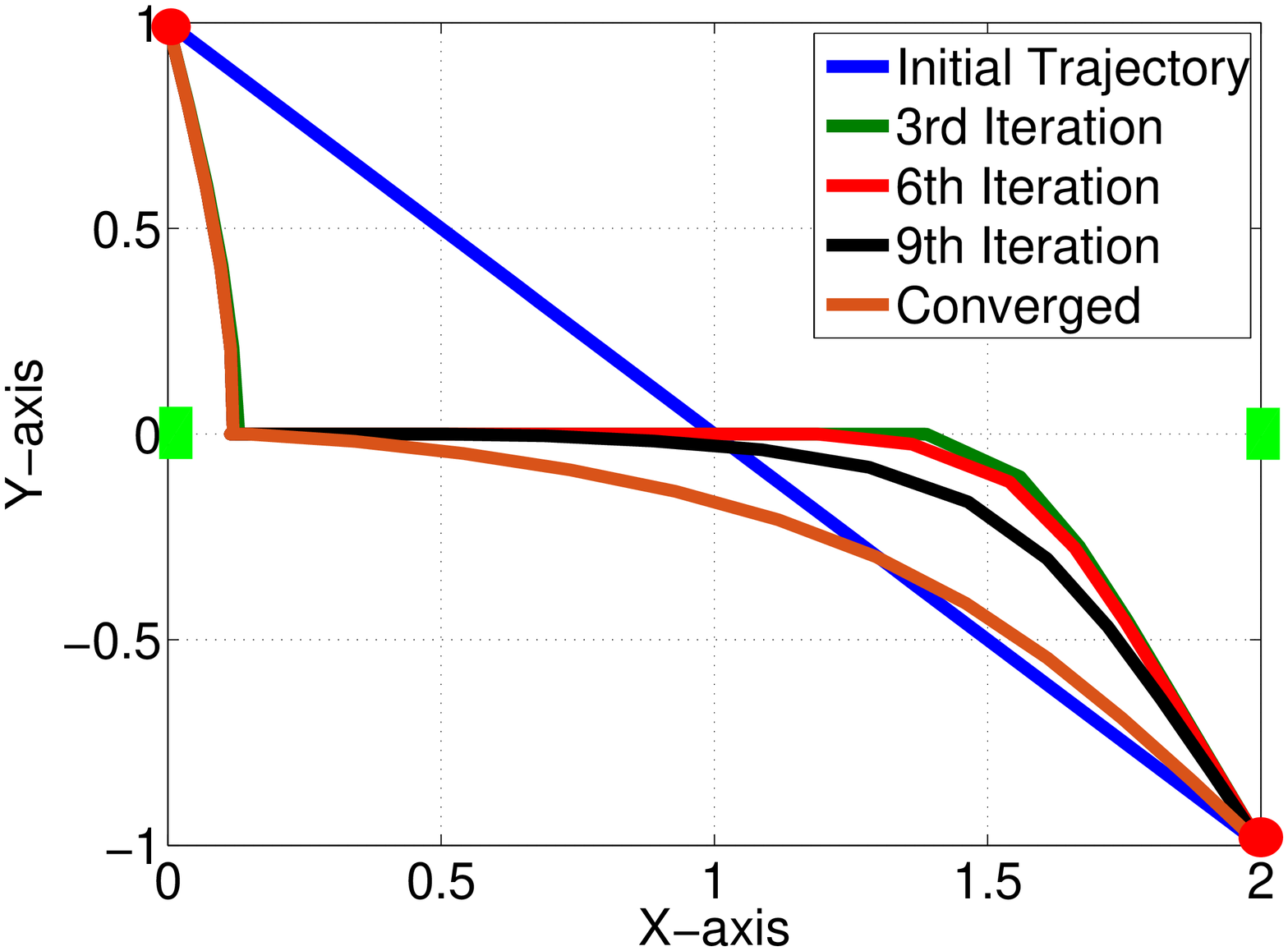}
\caption{The UAV's iteratively updated trajectory (projected) following Algorithm \ref{alg:trajectory}. The green rectangle are the source and destination and the red circles are the UAV's start and end locations.}
\label{fig:TrajectoryIteration}
\end{minipage}
\end{figure*}

We first plot the iterative evolution of cooperative throughput following Algorithm \ref{alg:trajectory}, as shown in Fig. \ref{fig:ThroughputIteration}, where the exact cooperative throughput and its lower bound defined by (\ref{eqt:lb_af}) and (\ref{eqt:lb_df}) are presented for both AF and DF protocols. It can be shown that given the UAV's power profile and power-splitting ratio profile, the cooperative throughput can be improved significantly via the successive convex optimization. Besides, Algorithm \ref{alg:trajectory} is able to converge just after a few rounds of iteration\footnote{The convergence rate of Algorithm \ref{alg:trajectory} is significantly affected by stop criterion, e.g., the minimum improvement for the lower bound in each iteration, which is always $0.1\%$ in our experiments.} and with current parameter setting, it converges faster with the DF protocol, which is resulted from the piecewise minimum in (\ref{eqt:df}).

Fig. \ref{fig:TrajectoryIteration} shows the UAV's iteratively updated trajectory following Algorithm \ref{alg:trajectory}, which are projected onto the ground plane for ease of demonstration. We can see that instead of following the initial straight-line trajectory, the UAV tends to fly following a trajectory that is close either to the source or the destination to gain better channel condition. Noticeably, it can be inferred from the evolution of the iteratively updated trajectory, that the UAV tends to find a shorter path that is closer to the source. This also indicates that channel condition of the UAV-to-destination link is less important than that of the source-to-UAV link.

\def\2Win{4.3in}
\begin{figure*}
\centering
\begin{minipage}[t]{\2Win}
\includegraphics[width=\2Win]{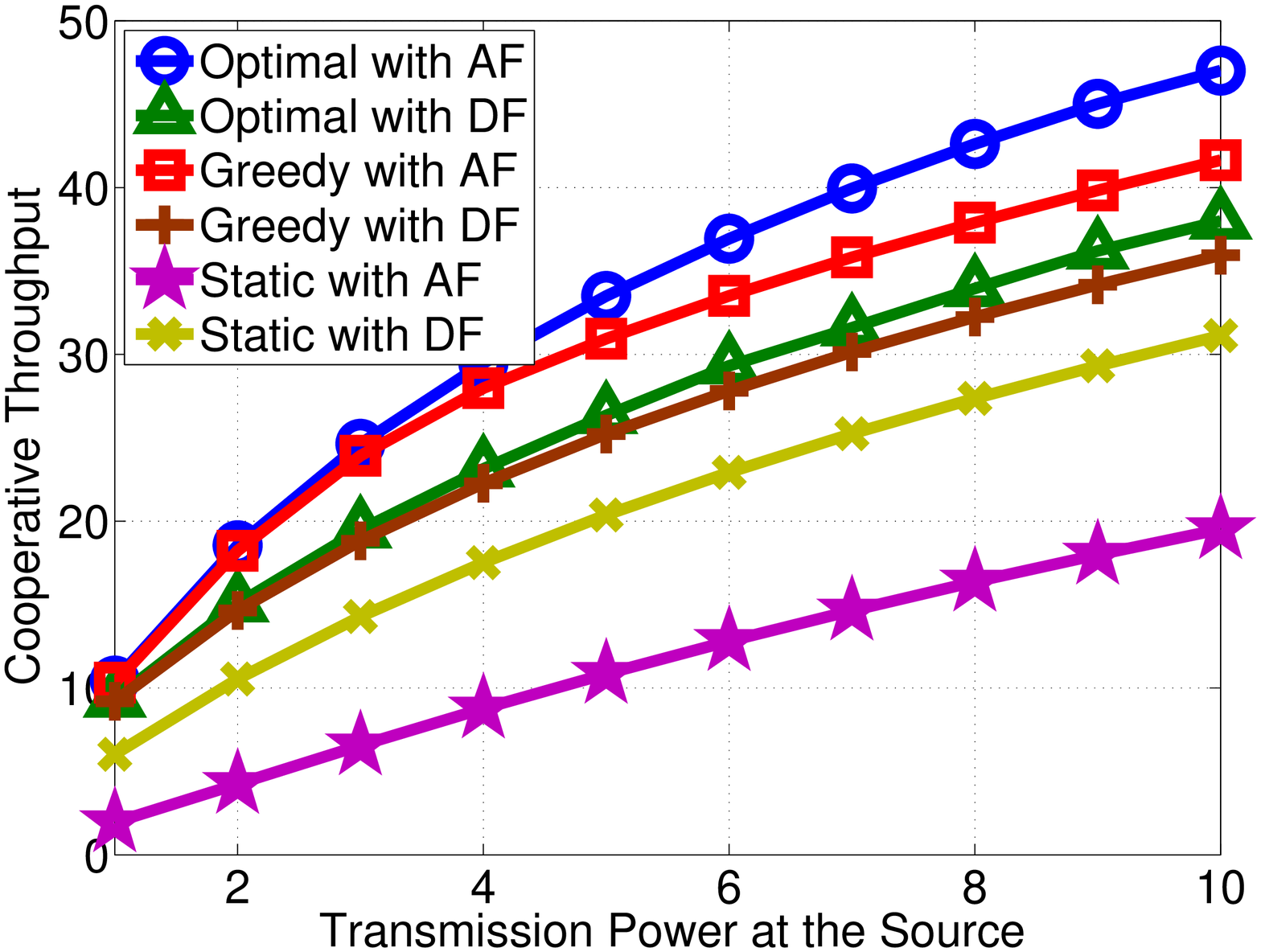}
\caption{Cooperative throughput comparison versus tranmission power at the source.}
\label{fig:RateComp}
\end{minipage}
\begin{minipage}[t]{\2Win}
\includegraphics[width=\2Win]{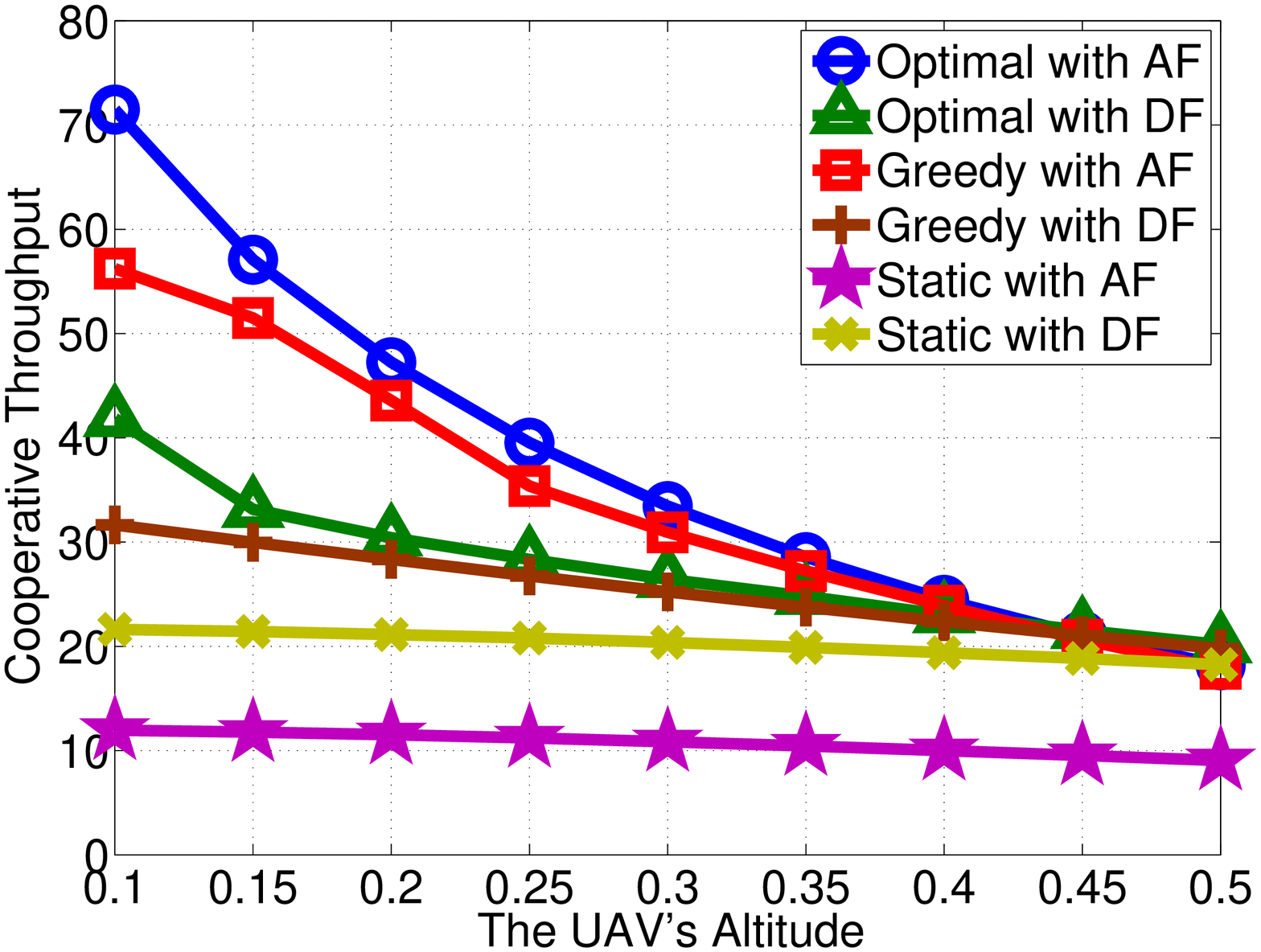}
\caption{Cooperative throughput comparison versus the UAV's altitude.}
\label{fig:RateCompHeight}
\end{minipage}
\end{figure*}

\subsection{Joint Profile and Trajectory Optimization}
We show the end-to-end cooperative throughput versus the transmission power at the source for both AF and DF protocols in Fig. \ref{fig:RateComp}, where a static strategy, i.e., the UAV is hovering in the air constantly at location $[0,1]$ (midpoint between the source and destination) for the entire $N$ time slots, is also included for comparison. Here the UAV's trajectory is initialized as a semi-circle for Algorithm \ref{alg:alternate_optimize}. It can be observed that for both AF and DF protocols, the static strategy is outperformed by both the proposed optimal solution and the greedy strategy and the performance gap is amplified as the transmission power at the source increases. This illustrates the non-negligible advantage of the UAV's high mobility over a conventional static relay. We can also find that the proposed optimal solution always outperforms the greedy strategy and the performance gap is amplified with higher transmission power at the source as well. This is because with the proposed optimal solution, energy carried by signal from the source is used more appropriately via the optimal power profile and power-splitting ratio profile. In addition, it is notable that the DF protocol outperforms the AF protocol for the static strategy while the opposite happens for the proposed optimal solution and the greedy strategy, which indicates that the AF protocol can have higher performance gain than the DF protocol with either of the two.

Furthermore, the end-to-end cooperative throughput versus the UAV's altitude is shown in Fig. \ref{fig:RateCompHeight}. Similar with Fig. \ref{fig:RateComp}, both the proposed optimal solution and the greedy strategy outperform the static strategy. However, the performance gain tends to shrink as the UAV flies at a higher altitude, which is resulted from worse channel condition of the cooperative link. In addition, it can be also observed that with excessively weak channel condition of the cooperative link, the greedy strategy almost presents equal performance compared to the proposed optimal solution. Therefore, when the cooperative link is not sufficiently favorable, the greedy strategy is more practicable to balance performance and computational complexity since it involves only a closed-form solution to the UAV's power profile and power-splitting ratio profile.

\section{Conclusion}
In this paper, we study the end-to-end cooperative throughput maximization problem for a typical cooperative communication system, where a UAV serves as an aerial mobile relay with both AF and DF protocols and its transmission capability is powered by radio signal transmitted from the source via the power-splitting scheme. By proving convexity of the problem, the UAV's power profile and power-splitting ratio profile given trajectory is optimized via dual decomposition. With fixed  power profile and power-splitting ratio profile, the UAV's trajectory is optimized via successive convex optimization, which iteratively optimizes the incremental trajectory to maximize a lower bound. Then the end-to-end cooperative throughput is maximized by alternately optimizing the UAV's profile and trajectory. Numerical results show that the proposed optimal solution makes a better choice for the UAV's profile and trajectory and outperforms the greedy and static strategies, especially for the AF protocol.


%

\appendices
\section{Proof of Theorem \ref{theorem:concavity}}\label{app:a}
For AF protocol, it is easy to find that $x_{n}(\rho_{n})=P_{s}\gamma_{0}\rho_{n}/[(\rho_{n} +a)(d^{sr}_{n})^{2}]$ is concave over $0\leq\rho_{n}\leq 1$ and $y_{n}(p_{n})=\gamma_{0}p_{n}/(d^{rd}_{n})^{2}$ is linear (affine) over $p_{n}\geq 0$. Moreover, $\log(1+P_{s}\gamma+x_{n}y_{n}/(1+x_{n}+y_{n}))$ is monotonically increasing with respect to both $x_{n}\geq 0$ and $y_{n}\geq 0$. Therefore, according to the concavity preserving theorems for composition and non-negative sum\cite{Boyd:2004:CO:993483}, to prove concavity of the objective function in (\ref{eqt:problem_profile}) is equivalent to prove concavity of $f(x_{n},y_{n})=x_{n}y_{n}/(1+x_{n}+y_{n})$. The concavity of $f(x_{n},y_{n})$ can be judged by its Hessian matrix, which is given by
\begin{equation}
\bigtriangledown ^{2}f(x_{n},y_{n})=
\left[
\begin{array}{cc}
\dfrac{-2y_{n}(y_{n}+1)}{(1+x_{n}+y_{n})^{3}} & \dfrac{2x_{n}y_{n}}{(1+x_{n}+y_{n})^{3}}\\ 
\dfrac{2x_{n}y_{n}}{(1+x_{n}+y_{n})^{3}} & \dfrac{-2x_{n}(x_{n}+1)}{(1+x_{n}+y_{n})^{3}}
\end{array}
\right]
\end{equation}
Obviously, we have $-2y_{n}(y_{n}+1){(1+x_{n}+y_{n})^{2}}<0$ and $|\bigtriangledown ^{2}f(x_{n},y_{n})|\geq 0$ for $x_{n}, y_{n}\geq 0$, which indicate that $\bigtriangledown ^{2}f(x_{n},y_{n})$ is non-positive definite. Therefore, the objective function in (\ref{eqt:problem_profile}) is concave over $p_{n}$ and $\rho_{n}$ for AF protocol. 

For DF protocol, it is evident that the two terms in (\ref{eqt:df}) are concave with respect to $\rho_{n}$ and $p_{n}$, respectively. Since pointwise minimum preserves concavity\cite{Boyd:2004:CO:993483}, the objective function in (\ref{eqt:problem_profile}) is concave over $p_{n}$ and $\rho_{n}$ for DF protocol.

Finally, it can be concluded that the objective function in (\ref{eqt:problem_profile}) is concave for both AF and DF protocols.

\section{Proof of (\ref{eqt:p_af}) and (\ref{eqt:rho_af})}\label{app:b}
Given $p^{n}$, we first define intermediate variables $c_{n,i}$ for $i=1,...,5$ as
\begin{equation}
\left\{
\begin{aligned}
&c_{n,1}=1+P_{s}\gamma+\dfrac{P_{s}p_{n}\gamma_{0}^{2}}{c_{n,2}(d^{sr}_{n}d^{rd}_{n})^{2}}\\
&c_{n,2}=1+\dfrac{P_{s}\gamma_{0}}{(d^{sr}_{n})^{2}}+\dfrac{p_{n}\gamma_{0}}{(d^{rd}_{n})^{2}}\\
&c_{n,3}=(1+\dfrac{p_{n}\gamma_{0}}{(d^{rd}_{n})^{2}})a\\
&c_{n,4}=(c_{n,1}-1-P_{s}\gamma)c_{n,3}\\
&c_{n,5}=\alpha_{n}(1+\dfrac{P_{s}\gamma_{0}}{(d^{sr}_{n})^{2}})\delta
\end{aligned}
\right.
\end{equation}
to facilitate the following derivation and it is easy to find that $c_{n,i}>0$ for $i=1,...,5$. The stationary point $\rho^{c}_{n}$ can be derived by solving $\partial f/\partial\rho_{n}=0$ as follows:
\begin{equation}\label{eqt:der_rho}
\begin{aligned}
&\dfrac{\partial f}{\partial\rho_{n}}=0\\
\Rightarrow &\dfrac{\partial r_{n}}{\partial\rho_{n}}-c_{n,5}=0\\
\Rightarrow &\dfrac{c_{n,2}c_{n,4}}{c_{n,1}(c_{n,2}\rho_{n}+c_{n,3})^{2}-c_{n,4}(c_{n,2}\rho_{n}+c_{n,3})}=2c_{n,5}\\
\Rightarrow &c_{n,1}\rho^{\prime 2}_{n}-c_{n,4}\rho^{\prime}_{n}-\dfrac{c_{n,2}c_{n,4}}{2c_{n,5}}=0
\end{aligned}
\end{equation}
where $\rho^{\prime}_{n}=c_{n,2}\rho_{n}+c_{n,3}$. The last equation in (\ref{eqt:der_rho}) indicates that there is one pair of positive and negative roots for possible stationary points. Obviously, we have $c_{n,3}\leq\rho^{\prime}_{n}\leq c_{n,2}+c_{n,3}$ for a feasible stationary point since $0\leq\rho_{n}\leq 1$. Therefore, the negative root is never a feasible stationary point and only the positive root is possibly feasible, which is given by
\begin{equation}
\rho^{c}_{n}=\dfrac{c_{n,4}+\sqrt{c_{n,4}^{2}+\dfrac{2c_{n,1}c_{n,2}c_{n,4}}{c_{n,5}}}}{2c_{n,1}c_{n,2}}-\dfrac{c_{n,3}}{c_{n,2}}
\end{equation}
Apparently, if $0\leq\rho^{c}_{n}\leq 1$, we have $\rho^{*}_{n}=\rho^{c}_{n}$ for the optimal power-splitting ratio profile given power profile. Otherwise, we have $\rho^{*}_{n}=\mathop{\argmax}_{\rho_{n}\in\{0,1\}}f(p_{n},\rho_{n})$, i.e., the better among $\rho_{n}=0$ and $\rho_{n}=1$ is selected as the optimal power-splitting ratio profile.

Given $\rho_{n}$, the optimal power profile can be derived following a similar way. By 
\begin{equation}
\left\{
\begin{aligned}
&b_{n,1}=1+P_{s}\gamma+\dfrac{P_{s}\gamma_{0}\rho_{n}}{(d^{sr}_{n})^{2}(\rho_{n}+a)}\\
&b_{n,2}=\dfrac{\gamma_{0}}{(d^{rd}_{n})^{2}}\\
&b_{n,3}=1+\dfrac{P_{s}\gamma_{0}\rho_{n}}{(d^{sr}_{n})^{2}(\rho_{n}+a)}\\
&b_{n,4}=\dfrac{P_{s}\gamma_{0}\rho_{n}b_{n,3}}{(d^{sr}_{n})^{2}(\rho_{n}+a)}\\
&b_{n,5}=\alpha_{n}
\end{aligned}
\right.
\end{equation}
The stationary point $p^{c}_{n}$ can be derived by solving $\partial f/\partial p_{n}=0$ as follows:
\begin{equation}\label{eqt:der_p}
\begin{aligned}
&\dfrac{\partial f}{\partial p_{n}}=0\\
\Rightarrow &\dfrac{\partial r_{n}}{\partial p_{n}}-b_{n,5}=0\\
\Rightarrow &\dfrac{b_{n,2}b_{n,4}}{b_{n,1}(b_{n,2}p_{n}+b_{n,3})^{2}-b_{n,4}(b_{n,2}p_{n}+b_{n,3})}=2b_{n,5}\\
\Rightarrow &b_{n,1}p^{\prime 2}_{n}-b_{n,4}p^{\prime}_{n}-\dfrac{b_{n,2}b_{n,4}}{2b_{n,5}}=0
\end{aligned}
\end{equation}
where $p^{\prime}_{n}=b_{n,2}p_{n}+b_{n,3}$. The last equation in (\ref{eqt:der_p}) indicates that there is one pair of positive and negative roots for possible stationary points. Obviously, we have $p^{\prime}_{n}\geq b_{n,3}$ for a feasible stationary point since $p_{n}\geq 0$. Therefore, the negative root is never a feasible stationary point and only the positive root is possibly feasible, which is given by
\begin{equation}
p^{c}_{n}=\dfrac{b_{n,4}+\sqrt{b_{n,4}^{2}+\dfrac{2b_{n,1}b_{n,2}b_{n,4}}{b_{n,5}}}}{2b_{n,1}b_{n,2}}-\dfrac{b_{n,3}}{b_{n,2}}
\end{equation}
Apparently, if $p^{c}_{n}\geq 0$, we have $p^{*}_{n}=p^{c}_{n}$ for the optimal power profile given power-splitting ratio profile. Otherwise, we have $p^{*}_{n}=0$.

Thus, (\ref{eqt:p_af}) and (\ref{eqt:rho_af}) are both proved.

\section{Proof of (\ref{eqt:p_df}) and (\ref{eqt:rho_df})}\label{app:c}
It is apparent that (\ref{eqt:candi_p_df}) can be readily derived by solving $\partial f_{1}/\partial p_{n}=0$, which has one unique solution. Therefore, taking $p_{n}\geq 0$ into account, we have (\ref{eqt:p_df}) as the optimal power profile given the dual variables.  

Similarly, by solving $\partial f_{2}/\partial \rho_{n}=0$, we have the following quadratic equation
\begin{equation}
\dfrac{2\alpha_{n}\delta}{\theta_{n}}(1+\dfrac{P_{s}\gamma_{0}}{(d^{sr}_{n})^{2}})^{2}(\rho_{n}+a)^{2}-
\dfrac{2a\alpha_{n}\delta P_{s}\gamma_{0}}{\theta_{n}(d^{sr}_{n})^{2}}(1+\dfrac{P_{s}\gamma_{0}}{(d^{sr}_{n})^{2}})(\rho_{n}+a)-\dfrac{aP_{s}\gamma_{0}}{(d^{sr}_{n})^{2}}=0
\end{equation}
which has two candidate roots. Since $0\leq\rho_{n}\leq 1$, only the root in (\ref{eqt:candi_rho_df}) could be feasible. Similar with Appendix \ref{app:b}, if $0\leq\rho^{c}_{n}\leq 1$, we have $\rho^{*}_{n}=\rho^{c}_{n}$ for the optimal power-splitting ratio profile. Otherwise, we have $\rho^{*}_{n}=\mathop{\argmax}_{\rho_{n}\in\{0,1\}}f_{2}(\rho_{n})$.

Thus, (\ref{eqt:p_df}) and (\ref{eqt:rho_df}) are both proved.

\section{Proof of Lemma \ref{lemma:convexity}}\label{app:d}
The convexity of (\ref{eqt:function_z}) can be validated through its Hessian matrix. By defining $X=(A_{1}+z_{1})(A_{2}+z_{2})+r_{1}(A_{2}+z_{2})+r_{2}(A_{1}+z_{1})$, we have
\begin{equation}\label{eqt:first_order}
\left\{
\begin{aligned}
&\dfrac{\partial f}{\partial z_{1}}=(A_{2}+z_{2}+r_{2})(\dfrac{1}{X+\dfrac{r_{1}r_{2}}{\omega}}-\dfrac{1}{X})\\
&\dfrac{\partial f}{\partial z_{2}}=(A_{1}+z_{1}+r_{1})(\dfrac{1}{X+\dfrac{r_{1}r_{2}}{\omega}}-\dfrac{1}{X})
\end{aligned}
\right.
\end{equation}
for the first-order derivatives.
Based on (\ref{eqt:first_order}), the second-order derivatives can be derived as
\begin{equation}\label{eqt:second_order}
\left\{
\begin{aligned}
&\dfrac{\partial^{2}f}{\partial z_{1}^{2}}=(A_{2}+z_{2}+r_{2})^{2}(\dfrac{1}{X^{2}}-\dfrac{1}{(X+\dfrac{r_{1}r_{2}}{\omega})^{2}})\\
&\dfrac{\partial^{2}f}{\partial z_{2}^{2}}=(A_{1}+z_{1}+r_{1})^{2}(\dfrac{1}{X^{2}}-\dfrac{1}{(X+\dfrac{r_{1}r_{2}}{\omega})^{2}})\\
&\dfrac{\partial^{2}f}{\partial z_{1}\partial z_{2}}=(A_{1}+z_{1}+r_{1})(A_{2}+z_{2}+r_{2})(\dfrac{1}{X^{2}}-\dfrac{1}{(X+\dfrac{r_{1}r_{2}}{\omega})^{2}})
\end{aligned}
\right.
\end{equation}
Then the Hessian matrix is given by
\begin{equation}
\bigtriangledown ^{2}f(z_{1},z_{2})=(\dfrac{1}{X^{2}}-\dfrac{1}{(X+\dfrac{r_{1}r_{2}}{\omega})^{2}})
\left[
\begin{array}{cc}
(A_{2}+z_{2}+r_{2})^{2} & (A_{1}+z_{1}+r_{1})(A_{2}+z_{2}+r_{2})\\ 
(A_{1}+z_{1}+r_{1})(A_{2}+z_{2}+r_{2}) & (A_{1}+z_{1}+r_{1})^{2}
\end{array}
\right]
\end{equation}
which can be easily found semi-definite since $0<X<X+r_{1}r_{2}/\omega$ always holds. Therefore, (\ref{eqt:function_z}) is convex over $z_{1}$ and $z_{2}$.

\section{Proof of Theorem \ref{theorem:lb_af}}\label{app:e}
Lemma \ref{lemma:convexity} has shown that (\ref{eqt:function_z}) is convex over $\bm{z}=[z_{1},z_{2}]^{T}$. Since the first-order Taylor approximation of a convex function is a global under-estimator\cite{Boyd:2004:CO:993483}, we have $f(\bm{z})\geq f(\bm{z}_{\bm{0}})+\bigtriangledown f(\bm{z}_{\bm{0}})^{T}(\bm{z}-\bm{z}_{\bm{0}})$ for any given $\bm{z}_{\bm{0}}=[z_{1,0},z_{2,0}]^{T}$. With $\bm{z}_{\bm{0}}=[0,0]^{T}$, we have the following inequality:
\begin{equation}
\begin{aligned}
\log(\omega +\dfrac{\dfrac{r_{1}r_{2}}{(A_{1}+z_{1})(A_{2}+z_{2})}}{1+\dfrac{r_{1}}{A_{1}+z_{1}}+\dfrac{r_{2}}{A_{2}+z_{2}}})&\geq \log(\omega +\dfrac{\dfrac{r_{1}r_{2}}{A_{1}A_{2}}}{1+\dfrac{r_{1}}{A_{1}}+\dfrac{r_{2}}{A_{2}}})\\
&-(\dfrac{1}{X_{0}}-\dfrac{1}{X_{0}+\dfrac{r_{1}r_{2}}{\omega}})[(A_{2}+r_{2})z_{1}+(A_{1}+r_{1})z_{2}]
\end{aligned}
\end{equation}
where $X_{0}=A_{1}A_{2}+r_{1}A_{2}+r_{2}A_{1}$.

By reviewing (\ref{eqt:af}), it can be found that $(d^{sr}_{n,l+1})^{2}=(d^{sr}_{n,l})^{2}+\Delta^{sr}_{n}$ and $(d^{rd}_{n,l+1})^{2}=(d^{rd}_{n,l})^{2}+\Delta^{rd}_{n}$ after $l$ rounds of iteration, where $\Delta^{sr}_{n,l}$ and $\Delta^{rd}_{n,l}$ are defined as $\Delta^{sr}_{n,l}=\Delta x_{n,l}^{2}+\Delta y_{n,l}^{2}+2(x_{n,l}-S_{x})\Delta x_{n,l}+2(y_{n,l}-S_{y})\Delta y_{n,l}$ and $\Delta^{rd}_{n,l}=\Delta x_{n,l}^{2}+\Delta y_{n,l}^{2}+2(x_{n,l}-D_{x})\Delta x_{n,l}+2(y_{n,l}-D_{y})\Delta y_{n,l}$, respectively. Then by letting $\omega=1+P_{s}\gamma$, $r_{1}=P_{s}\gamma_{0}\rho_{n}/(\rho_{n}+a)$, $r_{2}=p_{n}\gamma_{0}$, $A_{1}=(d^{sr}_{n,l})^{2}$, $A_{2}=(d^{rd}_{n,l})^{2}$, $z_{1}=\Delta^{sr}_{n}$ and $z_{2}=\Delta^{rd}_{n}$, (\ref{eqt:lb_af}) can be proved by defining
\begin{equation}\label{eqt:def}
\left\{
\begin{aligned}
&X_{n,l}=(d^{sr}_{n,l}d^{rd}_{n,l})^{2}+\dfrac{P_{s}\gamma_{0}\rho_{n}(d^{rd}_{n,l})^{2}}{(\rho_{n}+a)}+p_{n}\gamma_{0}(d^{sr}_{n,l})^{2}\\
&\phi_{n,l}=\dfrac{1}{X_{n,l}}-\dfrac{1}{X_{n,l}+\dfrac{P_{s}\gamma_{0}^{2}\rho_{n}p_{n}}{(\rho_{n}+a)(1+P_{s}\gamma)}}\\
&\mu_{n,l}=\phi_{n,l}[(d^{sr}_{n,l})^{2}+(d^{rd}_{n,l})^{2}+\dfrac{P_{s}\gamma_{0}\rho_{n}}{\rho_{n}+a}+p_{n}\gamma_{0}]\\
&\delta_{n,l}=2\phi_{n,l}[((d^{rd}_{n,l})^{2}+p_{n}\gamma_{0})(x_{n,l}-S_{x})+((d^{sr}_{n,l})^{2}+\dfrac{P_{s}\gamma_{0}\rho_{n}}{\rho_{n}+a})(x_{n,l}-D_{x})]\\
&\eta_{n,l}=2\phi_{n,l}[((d^{rd}_{n,l})^{2}+p_{n}\gamma_{0})(y_{n,l}-S_{y})+((d^{sr}_{n,l})^{2}+\dfrac{P_{s}\gamma_{0}\rho_{n}}{\rho_{n}+a})(y_{n,l}-D_{y})]
\end{aligned}
\right.
\end{equation}

%
%

\ifCLASSOPTIONcaptionsoff
  \newpage
\fi

\vspace{-2cm}

\end{document}